\documentclass[12pt]{article}

\usepackage[margin=1in]{geometry}

\usepackage[numbers]{natbib}
\usepackage[utf8]{inputenc} 
\usepackage[T1]{fontenc}    
\usepackage{hyperref}       
\usepackage{url}            
\usepackage{booktabs}       
\usepackage{multirow}
\usepackage{amsfonts}       
\usepackage{nicefrac}       
\usepackage{microtype}      
\usepackage{xcolor}         
\usepackage{amsmath}
\usepackage{amsthm}
\usepackage{graphicx}
\newcommand{\n}{n}
\newcommand{\nln}[1]{\n_{#1}}
\newcommand{\ellsub}[1]{\ell_{#1}}

\newcommand{\m}{m}
\newcommand{\mln}[1]{\m_{#1}}

\newcommand{\x}{x}
\newcommand{\y}{y}
\newcommand{\xln}[1]{\x_{#1}}
\newcommand{\xhln}[1]{\tilde{\x}_{#1}}

\newcommand{\yln}[1]{\y_{#1}}
\newcommand{\yhln}[1]{\tilde{\y}_{#1}}

\newcommand{\bx}{X}
\newcommand{\by}{Y}
\newcommand{\bz}{Z}
\newcommand{\mmdu}[2]{{\text{MMD}}^2_{u}(#1, #2)}

\newcommand{\mmdulower}[2]{{\underline{\text{MMD}}}^2_{u}(#1, #2)}
\newcommand{\mmduupper}[2]{{\overline{\text{MMD}}}^2_{u}(#1, #2)}
\newcommand{\dd}{d}
\newcommand{\ddimension}{\mathbb{R}^{\dd}}
\newcommand{\kernel}[2]{k(#1,#2)}
\newcommand{\kernels}[2]{k^*(#1,#2)}

\newcommand{\z}{z}
\newcommand{\zln}[1]{\z_{#1}}
\newcommand{\zsln}[1]{\z^*_{#1}}
\newcommand{\kk}{k} 
\newcommand{\st}{t}
\newcommand{\iconstant}{i}
\newcommand{\jconstant}{j}
\newcommand{\lconstant}{l}
\newcommand{\constant}[1]{c_{#1}}
\newcommand{\bt}{T}
\newcommand{\btln}[1]{\bt_{#1}}
\newcommand{\bs}{S}
\newcommand{\bsln}[1]{\bs_{#1}}

\newcommand{\br}{R}
\newcommand{\sa}{a}
\newcommand{\saln}[1]{\sa_{#1}}
\newcommand{\salnp}[1]{\sa_{#1}'}
\newcommand{\bb}{b}
\newcommand{\bbln}[1]{\bb_{#1}}
\newcommand{\bblnp}[1]{\bb_{#1}'}

\newcommand{\bu}{U}
\newcommand{\buln}[1]{\bu_{#1}}
\newcommand{\tonumber}[1]{[#1]}
\newcommand{\indicator}[1]{\text{I}(#1)}
\newcommand{\p}{p}

\newcommand{\upperp}{\overline{\p}}
\newcommand{\q}{q}
\newcommand{\nh}{H_0}
\newcommand{\ah}{H_1}
\newcommand{\pr}[1]{P\left(#1\right)}
\newcommand{\iid}{\text{i.i.d.}}
\newcommand{\termone}{A_1}
\newcommand{\termtwo}{A_2}
\newcommand{\termthree}{A_3}
\newcommand{\termfour}{A_4}
\newcommand{\setchi}{\mathcal{X}}
\newcommand{\setchiln}[1]{\setchi_{#1}}
\DeclareMathOperator*{\argmin}{arg\,min}
\newtheorem{definition}{Definition}
\newtheorem{lemma}{Lemma}

\newtheorem{theorem}{Theorem}

\title{MMD Two-sample Testing in the Presence \\ of Arbitrarily Missing Data}

\author{Yijin Zeng  \and Niall M. Adams \and Dean A. Bodenham}

\date{{\normalsize Department of Mathematics, Imperial College London,} \\
 {\normalsize South Kensington Campus, London SW7 2AZ, U.K.} \\
{\small \texttt{yijin.zeng20@imperial.ac.uk}, 
\quad \texttt{n.adams@imperial.ac.uk},
\quad \texttt{dean.bodenham@imperial.ac.uk}}
}

\begin{document}

\maketitle

\begin{abstract}
In many real-world applications, it is common that a proportion of the data may be missing or only partially observed. We develop a novel two-sample testing method based on the Maximum Mean Discrepancy (MMD) which accounts for missing data in both samples, without making assumptions about the missingness mechanism. Our approach is based on deriving the mathematically precise bounds of the MMD test statistic after accounting for all possible missing values. To the best of our knowledge, it is the only two-sample testing method that is guaranteed to control the Type I error for both univariate and multivariate data where data may be arbitrarily missing. Simulation results show that our method has good statistical power, typically for cases where 5\% to 10\% of the data are missing. We highlight the value of our approach when the data are missing not at random, a context in which either ignoring the missing values or using common imputation methods may not control the Type I error.
\end{abstract}


\section{Introduction} \label{introduction}

Two-sample hypothesis testing is a fundamental statistical method used to determine if 
two samples of data are different. 
Numerous two-sample testing methods are available, including Student's $t$-test 
\cite{student1908probable}, the Wilcoxon-Mann-Whitney $U$ 
test \cite{mann1947test}, the Kolmogorov-Smirnov test \cite{an1933sulla}, 
the Energy Distance \cite{rizzo2016energy}, and the Maximum Mean Discrepancy (MMD) test 
\cite{gretton2012kernel}. These methods have proven useful in various fields 
including medicine \cite{tsiatis2008covariate,o1979multiple}, finance 
\cite{lukic2024novel,jobson1981performance}, psychology \cite{pfister2013confidence}
and machine learning \cite{gretton2012kernel, gretton2009fast,schrab2023mmd}.
Nearly all two-sample testing methods are designed solely for samples that 
are fully observed. However, in many real-world datasets, 
a subset of univariate values may
be missing or multivariate values may only be partially observed.

When data are missing, common practices are either to ignore all 
missing values or impute these values using some imputation scheme.
Following this step, the data is treated as complete, allowing 
any standard two-sample testing method to be used. However, except 
in special cases, such as when the missing data are missing completely 
at random \cite{rubin1976inference}, these practices are often invalid 
as they risk increasing the probability of a Type I error occurring. 
Under certain assumptions, such as when the data are missing completely 
at random, missing at random, or when the data are missing not at random
but the missingness mechanisms (the samples can be subject to different
missingness mechanisms) can be explicitly specified 
\cite{schafer2002missing}, then certain sophisticated missing data 
methods exist, such as the expectation-maximization 
algorithm and the multiple imputation method \cite{schafer1999multiple}.
However, without these assumptions, relying on these methods is fraught 
with risk.

Recently a two-sample testing method for data with univariate missing values 
was proposed which does not rely on a specific imputation procedure 
\cite{zeng2024two}.
This method, based the Wilcoxon-Mann-Whitney test, rejects the 
null hypothesis after accounting for all possible values of the missing data 
by deriving bounds for the $U$ statistic. 
This approach avoids ignoring or imputing the missing data and is shown to 
control the Type I error without making assumptions about the missingness 
mechanisms, while also having good statistical power when 
the proportion of missing data is around $10\%$.
However, this approach is restricted to univariate samples and can only detect 
location shifts since it is based on the Wilcoxon-Mann-Whitney test. 

In this paper, we propose a novel two-sample testing method based 
on the MMD test in the presence of missing data that makes no assumption about 
the missingness mechanisms, following the approach in \cite{zeng2024two}. 
However, 
our approach can be used for both univariate and multivariate samples, 
and can detect any distributional change since it is based on the MMD test.

Our approach, named \textbf{MMD-Miss}, 
is based on deriving the bounds of the unbiased MMD test 
statistic \cite{gretton2012kernel} in the presence of missing data. 
To do so, we use the Laplacian kernel, a popular choice of characteristic 
kernel that enables MMD to detect any distributional shift 
\cite{sriperumbudur2011universality}. 
To compute the $p$-value for the test, we either use 
Monte Carlo sampling from the permuation distribution of the observed data or, 
in the case of high-dimensional multivariate data, we use the normality 
approximation recently derived in \cite{gao2023two}.
We prove that this approach controls the Type I error, 
regardless of the values of the missing data. Numerical simulations are 
presented, showing that our method has good testing power when the 
proportion of missing data is around 5\% to 10\%, while controlling
the Type I error.

\paragraph{Two-sample testing for a small proportion of missing data} 
In the context of performing two-sample testing with missing data, a common 
misunderstanding is that when the proportion of missing data is small, 
e.g. 5\% as suggested by \cite{schafer1999multiple, heymans2022handling}, 
the testing result will not be skewed by missing data after either ignoring 
or imputing these values. While these practices may be justified under certain 
conditions, such as when the data are missing completely at 
random \cite{rubin1976inference}, we provide experiments in 
Section \ref{Experiments} showing that these practices can lead to a Type I error 
asymptotically equals to 1 even with a small proportion (5\%) 
of missing data, when the data are missing not at random.

\section{Background} \label{Background}

\paragraph{Related Work}

Various approaches have been proposed to mitigate the issue of missing data for 
two-sample hypothesis testing. When the missing data are missing at random, 
rank-based two-sample testing methods have been proposed which control the Type I error 
and achieve good testing power \cite{cheung2005exact,chen2013mann}. 
Paired two-sample testing with missing data are studied in 
\cite{martinez2013hypothesis, fong2018rank}, under the assumption that data are missing 
completely at random and at least one of the paired samples is observed. When 
part of the samples are interval-censored, i.e. the samples are 
bounded within an interval, \cite{pan2000two} proposed an 
algorithm based on multiple imputation \cite{little2019statistical} by 
employing approximate Bayesian bootstrap \cite{rubin1986multiple}. 
A method for detecting a change in the mean 
of high-dimensional normal data 
assumes the two samples share the same missingness mechanism 
\cite{shutoh2010testing}. 
By considering certain unobserved samples as 
the ``worst case'',  \cite{lachin1999worst} proposes a test method based on 
rank called the worst rank test.


\paragraph{Missingness Mechanisms.} Missingness mechanisms are typically 
classified as missing completely at random (MCAR), missing at random 
(MAR), and missing not at random (MNAR) \cite{rubin1976inference}. 
Let $\bz_{\text{com}} = (\zln{1}, \ldots, \zln{\n})^T$ denote the 
complete data, where each $\zln{\iconstant}$ is a $\mathbb{R}^d$ real 
value. The complete data $\bz_{\text{com}}$ can be split into observed 
parts $\bz_{\text{obs}}$ and missing parts $\bz_{\text{miss}}$. 
Let $\br$ be a binary indicator matrix matching the dimensions 
of $\bz_{\text{com}}$, with elements indicating observed (1) or missing 
(0) values. The core idea of \cite{rubin1976inference} is to 
admit $\br$ as a probabilistic phenomenon.

The mechanism is MCAR if $\pr{\br|\bz_{\text{com}}} = \pr{\br}$, 
meaning $\br$ is independent of data values, justifying the practice of 
ignoring missing data. It is MAR 
if $\pr{\br|\bz_{\text{com}}} = \pr{\br|\bz_{\text{obs}}}$, where $\br$ 
depends only on observed data, making certain imputation methods viable. 
If neither condition holds, the mechanism is MNAR. For known mechanisms 
\cite{rubin1976inference, lachin1999worst}, specific imputations might be 
justified. However, with unknown mechanisms, ignoring or imputing data can 
lead to invalid results, as shown in our simulations 
(Section \ref{Experiments}). Notably, our method does not depend on the 
missing data values, ensuring valid results regardless of the missingness 
mechanism.

\paragraph{Two-Sample Testing.} A two-sample testing method is used to determine 
whether two groups of data are statistically significantly different.
Consider two samples of observations
$\bx = \{\xln{1}, \ldots, \xln{\nln{1}}\}$ and 
$\by = \{\yln{1}, \ldots, \yln{\nln{2}}\}$, where each observation is in  
$\ddimension$, and suppose that $\bx$, $\by$ are independent random samples 
from distributions $\p$, $\q$, respectively. 
We define the null hypothesis $\nh$ as $\p = \q$ and the alternative 
hypothesis $\ah$ as $\p \neq \q$. 
The two-sample testing method first computes a statistic based on the data, 
and then a $p$-value based on the statistic. After comparing the $p$-value
to a pre-specified significance threshold $\alpha$, the null hypothesis $\nh$ 
is either rejected or fails to be rejected.

A Type I error occurs when the null hypothesis $\nh$ is true but the test 
incorrectly rejects it. Conversely, a Type II error occurs when $\nh$ is 
false but the test fails to reject it. A two-sample testing method is usually 
derived so that, under $\nh$, the Type I error is not larger than  
the pre-specified significance threshold $\alpha \in (0,1)$. 
In this case, we say the test controls the Type I error.
For a given $\alpha$, a 
two-sample testing method is preferred if it has a lower probability of making 
a Type II error, denoted by $\beta$. The power is defined as 
$1 - \beta$, and a two-sample testing method is more desirable if it 
demonstrates higher power. Two-sample testing methods are usually assessed
on their power, given that the Type I error is controlled.

\paragraph{Maximum Mean Discrepancy.} The Maximum Mean Discrepancy (MMD) is a 
kernel-based measure of distance between two distributions $\p$ and $\q$, 
by comparing their mean embeddings in a reproducing kernel 
Hilbert space \cite{aronszajn1950theory}. 
More formally, let $\mathcal{H}_k$ denote a reproducing kernel Hilbert space 
with kernel function $\kk$, the MMD between $\p$ and $\q$ is defined 
as \cite{gretton2012kernel}
\begin{align*}
	\text{MMD}^2[\mathcal{H}_k, \p, \q] = \sup_{f \in \mathcal{H}_k: ||f|| \le 1} (\mathbb{E}_{\bx \sim \p}[f(\bx)] - \mathbb{E}_{\by \sim \q}[f(\by)]).
\end{align*}
In the context of two-sample testing, an unbiased estimate of $\text{MMD}^2[\mathcal{H}_k, \p, \q]$ is:
\begin{align} 
	\begin{split}
		\mmdu{\bx}{\by} &= \tfrac{1}{\nln{1}(\nln{1}-1)} \sum_{\iconstant = 1}^{\nln{1}} \sum_{ \substack{\jconstant = 1 \\ \jconstant \neq \iconstant } }^{\nln{2}} \kernel{\xln{\iconstant}}{\xln{\jconstant}}
		+ \tfrac{1}{\nln{2}(\nln{2}-1)} \sum_{\iconstant = 1}^{\nln{2}} \sum_{ \substack{\jconstant = 1 \\ \jconstant \neq \iconstant } }^{\nln{2}} \kernel{\yln{\iconstant}}{\yln{\jconstant}}
		- \tfrac{2}{\nln{1}\nln{2}}	\sum_{\iconstant = 1}^{\nln{1}} \sum_{\jconstant = 1}^{\nln{2}} \kernel{\xln{\iconstant}}{\yln{\jconstant}}.
        \nonumber
	\end{split}
\end{align}
The MMD test rejects the null hypothesis when $\mmdu{\bx}{\by}$ exceeds 
certain threshold, usually chosen to control the probability of Type I 
error under a pre-specified number $\alpha \in (0,1)$. It is shown 
in \cite{gretton2012kernel} that when  $\kk$ is a characteristic kernel 
\cite{sriperumbudur2011universality}, $\text{MMD}^2[\mathcal{H}_k, \p, \q] =0$ 
if and only if $\p = \q$, otherwise $\text{MMD}^2[\mathcal{H}_k, \p, \q] > 0$. 
Hence, when a characteristic kernel is used, the MMD test is able to detect 
any distribution changes. However, on finite sample sizes, the choices of 
kernel have significant impact on the power of MMD test \cite{biggs2024mmd}. 
Two popular choices of characteristic kernels are Laplacian and Gaussian 
kernels, which are defined as 
\begin{align} 
    \kk_{L}({\x},{\y}) = \exp(-\beta ||\x - \y||_1),~~
	\kk_{G}({\x},{\y}) = \exp(-\gamma ||\x - \y||_2),
    \nonumber
\end{align}   
respectively, where $\gamma > 0$ and $\beta > 0$ are hyperparameters.

\paragraph{Permutation Test.} The permutation test is a common numerical 
procedure for approximating the distribution of $\mmdu{\bx}{\by}$ under 
the null hypothesis and deciding the test threshold. It is proved to be 
able to control the Type I error of MMD test non-asymptotically 
\cite{schrab2023mmd}. Let $\bx = \{\xln{1}, \ldots, \xln{\nln{1}}\}$ 
and $\by = \{\yln{1}, \ldots, \yln{\nln{2}}\}$, where each sample is 
$\ddimension$ real value. The permutation test begins with uniformly 
sampling $B$ $\iid$ permutations from $\{1,\ldots,\nln{1} + \nln{2}\}$. 
Let us denote the $B$ samplings as $(\sigma^{(1)}, \ldots, \sigma^{(B)})$, 
where each 
$\sigma^{(\iconstant)} = (\sigma^{(\iconstant)}(1),\ldots, 
\sigma^{(\iconstant)}(\nln{1}+\nln{2}))$, $\iconstant = 1,\ldots,B$. 
Subsequently, define 
$\zln{1} = \xln{1},\ldots, \zln{\nln{1}} = \xln{\nln{1}}, 
\zln{\nln{1}+1} = \yln{1},\ldots, \zln{\nln{1}+\nln{2}} = \yln{\nln{2}}$ 
and for any $\iconstant = 1,\ldots,B$, denote
\begin{align*}
	\bx_{\sigma^{(\iconstant)}} 
    = \{\zln{\sigma^{(\iconstant)}(1)},\ldots,\zln{\sigma^{(\iconstant)}(\nln{1})} \}, 
    \by_{\sigma^{(\iconstant)}} = \{\zln{\sigma^{(\iconstant)}(\nln{1}+1)},\ldots,
    \zln{\sigma^{(\iconstant)}(\nln{1}+\nln{2})} \}.
\end{align*}
The test threshold of MMD test, using $\mmdu{\bx}{\by}$ as test statistic 
with significance level $0 < \alpha < 1$, is then computed as 
the $\left\lceil \alpha(B+1) \right\rceil$-th largest numbers in the set
\begin{align*}
    \left\{\mmdu{\bx}{\by} \right\} \cup 
    \left\{ \mmdu{\bx_{\sigma^{(\iconstant)}}}{\by_{\sigma^{(\iconstant)}}}, 
    \iconstant = 1,\ldots,B\right\}.
\end{align*}
If $\mmdu{\bx}{\by}$ exceeds this threshold, the MMD test rejects the null 
hypothesis. Otherwise, the null hypothesis is not rejected.


\paragraph{Normality Approximation.}

It is well known that for fixed dimension $\dd$, the asymptotic 
distribution of $\mmdu{\bx}{\by}$ under the null hypothesis (i.e. $\p = \q$) 
takes the form of an infinite weighted sum of $\chi^2$ random variables 
with the weights depending on $\p$ \cite{gretton2012kernel}, which are 
normally unknown. An estimation of the asymptotic distribution of 
$\mmdu{\bx}{\by}$ is establish in \cite{gretton2009fast}, which is 
proved asymptotically correct regardless of the distribution $\p$ under 
certain conditions. However, empirical experiments show this estimation 
is rather conservative.

It was recently shown \cite{gao2023two} that the studentized 
$\mmdu{\bx}{\by}$ given by
Equation~\eqref{def:studentized mmd} convergences to 
normal distribution when all $\n,\m,\dd \to \infty$. This result holds 
for a wide range of kernels, including Laplacian and Gaussian kernels. 
This normality approximation is proved empirically rather accurate 
\cite{gao2023two} even for relatively small $\n, \m, \dd$, 
e.g. $\n = \m = 25, \dd = 50$, and the higher approximation accuracy 
is shown for larger $\n, \m, \dd$.

\section{Bounding MMD with Missing Data} \label{Bounding MMD with Missing Data}

In this section, we provide bounds for the MMD test statistic in the presence of 
missing data. To do so, we use the Laplacian kernel which is a characteristic 
kernel \cite{fukumizu2008characteristic} and able to detect any change in 
distribution. In the following, $\kernel{\x}{\y}$ will be used to denote the 
Laplacian kernel only. 

\paragraph{Univariate Samples.} 
If not all samples in $\bx$ and $\by$ are observed, 
the $\mmdu{\bx}{\by}$ statistic cannot be directly computed. 
However, for any two given real values $\x$ and $\y$, the Laplacian kernel 
$\kernel{\x}{\y} = \exp\left(-\beta |\x- \y|\right)$ is bounded by $(0,1]$ 
since $\beta > 0$ and $|\x - \y| > 0$. Hence, a straightforward way
to obtain a lower bound of $\mmdu{\bx}{\by}$ is to 
let $\kernel{\xln{\iconstant}}{\xln{\jconstant}} = 0$ if at least one sample 
of $\xln{\iconstant}$ and $\xln{\jconstant}$ is unobserved. Similarly, if at 
least one sample of $\yln{\iconstant}$ and $\yln{\jconstant}$ is unobserved, 
let $\kernel{\xln{\iconstant}}{\xln{\jconstant}} = 0$. On the other hand, 
to minimize $- \kernel{\xln{\iconstant}}{\yln{\jconstant}}$, we can 
let $- \kernel{\xln{\iconstant}}{\yln{\jconstant}} = -1$ if either 
$\xln{\iconstant}$ or $\yln{\jconstant}$ is unobserved. The upper 
bound of $\mmdu{\bx}{\by}$ can be obtained following the similar manner. 
This simplistic method, while providing an theoretically correct bound 
of $\mmdu{\bx}{\by}$, yields bounds that are too conservative for effective 
two-sample testing in practice. 
In the following, we construct tighter bounds of $\mmdu{\bx}{\by}$ in 
order to make it useful for the two-sample testing problem. 
One starts by decomposing $\mmdu{\bx}{\by}$ into terms which have either none, one or both
arguments of the kernel function missing; see Appendix \ref{supp:decomposition} for details. 
Given this decomposition, one only needs to bound
functions of the form in Equation~\eqref{eqn:boundT} below.
\begin{lemma} \label{lemma:2}
	Let $\xln{1}, \ldots, \xln{\ellsub{1}}$ and $\yln{1}, \ldots, \yln{\ellsub{2}}$ 
    be real values, that are observed.
    Suppose $\saln{1}, \ldots, \saln{\ellsub{1}}$, 
    $\bbln{1}, \ldots, \bbln{\ellsub{2}}$ and $\beta$ are positive constants. 
    Define
	\begin{align}
		\bt(\z) = \sum_{\iconstant = 1}^{\ellsub{1}} 
        \saln{\iconstant} \exp(-\beta|\xln{\iconstant} - \z|) 
        - \sum_{\iconstant = 1}^{\ellsub{2}} \bbln{\iconstant} \exp(-\beta|\yln{\iconstant} - \z|)
        \label{eqn:boundT}
	\end{align}
	as a function of $\z \in \mathbb{R}$. 
    Subsequently, for any given $\zln{0} \in \mathbb{R}$,
	\begin{align*} 
		\bt(\zln{0}) \ge \min \{0, \bt(\xln{1}),\ldots,  \bt(\xln{\ellsub{1}}), \bt(\yln{1}),\ldots,  \bt(\yln{\ellsub{2}})\},
	\end{align*}
	and
	\begin{align*} 
		\bt(\zln{0}) \le \max \{0, \bt(\xln{1}),\ldots,  \bt(\xln{\ellsub{1}}), \bt(\yln{1}),\ldots,  \bt(\yln{\ellsub{2}})\}.
	\end{align*}
\end{lemma}
This lemma is proved in Appendix \ref{Proof of Lemma 2}.
It provides a linear time algorithm for computing the bounds of $\bt(\z)$ 
by computing all the values in the set 
$\{ \bt(\z) : \z \in \{\xln{1}, \ldots, \xln{\ellsub{1}}, \yln{1}, \ldots, \yln{\ellsub{2}}\} \}$ 
and identifying the maximum and minimum values.

\paragraph{Multivariate Samples.} We now consider the case 
where $\bx = \{\xln{1}, \ldots, \xln{\nln{1}} \}$, 
$\by = \{\yln{1}, \ldots, \yln{\nln{2}}\}$ 
and $\bx,\by \in \mathbb{R}^{\dd}$ with $\dd > 1$. 
For each sample $\xln{\iconstant}$ in $\bx$, denote 
$\xln{\iconstant} = (\xln{\iconstant}(1), \ldots, \xln{\iconstant}(\dd))$, 
where $\xln{\iconstant}(\lconstant) \in \mathbb{R}$ is the value 
the $\lconstant$th component of $\xln{\iconstant}$. Similarly, for each 
sample $\yln{\iconstant}$ in $\by$, denote 
$\yln{\iconstant} = (\yln{\iconstant}(1), \ldots, \yln{\iconstant}(\dd))$. 
We now provide notation for which components of a $d$-dimensional
observation are missing and define what it means to impute such a value.
\begin{definition}
    For a value $\z = (\z(1), \dots, \z(d)) \in \mathbb{R}^{d}$, let 
    $\buln{\z}$ denote the set of components of $\z$ that are missing; 
    in other words $\{\z(j)  \,: \, j \in \buln{\z}\}$ are missing.
\end{definition}
\begin{definition}
    For a value $\z \in \mathbb{R}^{d}$, that has
    missing components, $z^{\ast}$ is called an \emph{imputation} of $z$ if
    $z^{\ast} \in \mathbb{R}^{d}$ 
    is fully observed and $z^{\ast}(j) = z(j)$ for all $j \in \{1, \dots, d \} \setminus \buln{z}$.
\end{definition}
We can now extend Lemma \ref{lemma:2} to the $d$-dimensional case:
\begin{lemma} \label{lemma:4}
	Let $\xln{1}, \ldots, \xln{\ellsub{1}}, \yln{1}, \ldots, \yln{\ellsub{2}} \in \mathbb{R}^d$
    be values that are fully observed.
    Suppose $\saln{1}, \ldots, \saln{\ellsub{1}}$, $\bbln{1}, \ldots, \bbln{\ellsub{2}}, \beta$ 
    are positive constants. For $\z = (\z(1), \ldots, \z(\dd) ) \in \mathbb{R}^d$  
    with missing components, define
    \begin{align*}
    	\bt( \{ \z (j) \,:\, j \in \buln{z} \}) 
    	= \sum_{\iconstant = 1}^{\ellsub{1}} \saln{\iconstant}  
    	\exp\left(-\beta \sum_{\jconstant \in \buln{\z}} | \xln{\iconstant}(\jconstant) 
    	- \z(\jconstant)|\right) 
    	- \sum_{\iconstant = 1}^{\ellsub{2}} \bbln{\iconstant} 
    	\exp\left(-\beta  \sum_{\jconstant \in \buln{\z}} |\yln{\iconstant}(\jconstant) - \z(\jconstant)|\right)
    \end{align*}
    as a function of the unobserved components of $\z$ and let 
    \begin{align*}
    	\setchi = \{ \bt( \{ \z(j) \,:\, j \in \buln{z} \}) \ :\, 
    	\z(\iconstant)\in \{\xln{1}(\iconstant), \ldots, \xln{\ellsub{1}}(\iconstant), 
    	\yln{1}(\iconstant), \ldots, \yln{\ellsub{2}}(\iconstant) \}, \iconstant \in \buln{z} \}.
    \end{align*}
    Then, for any possible imputation $z^{\ast}$ of $\z$, 
	\begin{align*}
        \min \{0, \min\setchi\} \leq \bt( \{ z^{\ast}(j) \,:\, j \in \buln{z} \})  \leq \max \{0, \max\setchi\} 
	\end{align*}
\end{lemma}
Lemma \ref{lemma:4} is proved in Appendix \ref{Proof of Lemma 4}. 
It shows that in order to compute the maximum and minimum
values of 
$\bt( \{ \z (j) \,:\, j \in \buln{z} \})$ for $z \in \mathbb{R}^d$,
we only need to check the imputations of $\z$ where its missing components
are imputed using the components of 
$\xln{1}, \ldots, \xln{\ellsub{1}}, \yln{1}, \ldots, \yln{\ellsub{2}}$.
However, computing $\bt( \{ \z (j) \,:\, j \in \buln{z} \})$ for all possible 
imputations using Lemma \ref{lemma:4} is 
$(\ellsub{1} + \ellsub{2})^{|\buln{\xln{\iconstant}}|}$, which 
is exponential in the number of unobserved components of $\z$ and impractical
to compute.
To address this computational challenge, we propose to further bound $\min\setchi$ and $\max\setchi$, using the following lemma:
\begin{lemma} \label{lemma:5}
	Following the notation and definitions in Lemma \ref{lemma:4}, denote 
	\begin{align*}
		\xhln{\iconstant}(\jconstant) &= \max \{ |\xln{\iconstant}(\jconstant) - \xln{1}(\jconstant)|, \ldots, |\xln{\iconstant}(\jconstant) - \xln{\ellsub{1}}(\jconstant)|, |\xln{\iconstant}(\jconstant) - \yln{1}(\jconstant)|, \ldots, |\xln{\iconstant}(\jconstant) - \yln{\ellsub{2}}(\jconstant)|\}
	\end{align*}
	for any $\iconstant \in \{1,\ldots,\ellsub{1}\}, \jconstant \in \buln{z}$; denote
	\begin{align*}
	    \yhln{\iconstant}(\jconstant) &= \max \{ |\yln{\iconstant}(\jconstant) - \ellsub{1}(\jconstant)|, \ldots, |\yln{\iconstant}(\jconstant) - \xln{\nln{1}}(\jconstant)|, |\yln{\iconstant}(\jconstant) - \yln{1}(\jconstant)|, \ldots, |\yln{\iconstant}(\jconstant) - \yln{\ellsub{2}}(\jconstant)|\}
	\end{align*}
	for any $\iconstant \in \{1,\ldots,\ellsub{2}\}, \jconstant \in \buln{z}$. Subsequently,
	\begin{align*}
		&\min \setchi \ge  \sum_{\iconstant = 1}^{\ellsub{1}} \saln{\iconstant} \exp\left(-\beta \sum_{ \jconstant \in \buln{z}} \xhln{\iconstant}(\jconstant) \right) - \sum_{\iconstant = 1}^{\ellsub{2}} \bbln{\iconstant},~\max \setchi \le  \sum_{\iconstant = 1}^{\ellsub{1}} \saln{\iconstant} - \sum_{\iconstant = 1}^{\ellsub{2}} \bbln{\iconstant} \exp\left(-\beta \sum_{ \jconstant \in \buln{z}} \yhln{\iconstant}(\jconstant) \right).
	\end{align*}
\end{lemma}
This lemma is proved in Appendix \ref{Proof of Lemma 5}. It can be seen that for each $\iconstant \in \{1,\ldots,\ellsub{1}\}, \jconstant \in \buln{z}$, $|\xhln{\iconstant}(\jconstant) | = \ellsub{1} + \ellsub{2}$ according to its definition. Hence, to bound $\min \setchi$, the computation times $|\buln{z}|\ellsub{1}(\ellsub{1} + \ellsub{2})$ are needed, which is of order ${O}(\dd \nln{1}(\nln{1} + \nln{2}))$. Similarly, to bound $\max \setchi$, the order of computation complexity is ${O}(\dd\nln{2}(\nln{1} + \nln{2}))$. This computational cost reduction, compared with Lemma \ref{lemma:4}, which grows exponentially with $|\buln{z}|$, makes bounding $\mmdu{\bx}{\by}$ possible in practice.

\paragraph{Bounds for MMD statistic}
The bounds derived above do not depend on the missing data mechanisms.
Combining Lemmas~\ref{lemma:2}, \ref{lemma:4} and \ref{lemma:5}, we have our main result:
\begin{theorem} \label{theorem:1}
    For data $\xln{1}, \ldots, \xln{\nln{1}} \in \mathbb{R}^d$ and 
    $\yln{1}, \ldots, \yln{\nln{2}} \in \mathbb{R}^d$ 
    with missing values, bounds for the MMD
    statistic with the Laplacian kernel 
    can be computed in $O(\nln{1} + \nln{2})$ when $d=1$ and  
    $O (d (\nln{1} + \nln{2})^2)$ when $d > 1$.
\end{theorem}
This theorem is proved by a series of results in Appendix~\ref{Proof of Theorem 1} and 
\ref{Proof of Theorem 2}.



\section{Two-Sample Testing in the Presence of Arbitrarily Missing Data} \label{Testing with Missing Data Using Bounds of MMD}

In this section, we discuss methods for employing bounds derived in 
Section~\ref{Bounding MMD with Missing Data} to develop valid two-sample 
testing methods. 
Besides computing the test statistic, we also need a method for computing 
a $p$-value from the test statistic.

\paragraph{Bounding $p$-value using permutations.} The permutation test, 
as discussed in Section \ref{Background}, is a numerical procedure that can 
be used to a $p$-value for any statistic based on the observed data. 
Instead of computing the threshold as in Section \ref{Background}, an 
equivalent approach to compute the $p$-value of the MMD test is 
\begin{align} \label{def:p value using permutations}
	\p = \frac{1}{B + 1} \left(1 + \sum_{\iconstant = 1}^{B} \indicator{ \mmdu{\bx_{\sigma^{(\iconstant)}}}{\by_{\sigma^{(\iconstant)}}} \ge \mmdu{\bx}{\by}} \right).
\end{align}
The null hypothesis is rejected if $\p$ is smaller or equal to the 
pre-specified significance level $\alpha$.  

Using the bounds of MMD in the presence of missing data, we proceed by 
providing bounds of the $\p$-value, using the following theorem: 
\begin{theorem} \label{theorem:3}
	Suppose $\bx = \{\xln{1}, \ldots, \xln{\nln{1}}\}, \by = \{\yln{1}, \ldots, \yln{\nln{2}}\} \in \ddimension$. 
    Let $(\sigma^{(1)}, \ldots, \sigma^{(B)})$ be $B$ $\iid$ 
    random permutations of $\{1,\ldots,\nln{1} + \nln{2}\}$ and denoted as $\sigma^{(\iconstant)} = (\sigma^{(\iconstant)}(1),\ldots, \sigma^{(\iconstant)}(\nln{1}+\nln{2}))$, $\iconstant = 1,\ldots,B$. Subsequently, let $\zln{1} = \xln{1},\ldots, \zln{\nln{1}} = \xln{\nln{1}}, \zln{\nln{1}+1} = \yln{1},\ldots, \zln{\nln{1}+\nln{2}} = \yln{\nln{2}}$ and for any $\iconstant = 1,\ldots,B$, denote
	\begin{align*}
		\bx_{\sigma^{(\iconstant)}} = \{\zln{\sigma^{(\iconstant)}(1)},\ldots,\zln{\sigma^{(\iconstant)}(\nln{1})} \}, \by_{\sigma^{(\iconstant)}} = \{\zln{\sigma^{(\iconstant)}(\nln{1}+1)},\ldots,\zln{\sigma^{(\iconstant)}(\nln{1}+\nln{2})} \},
	\end{align*}
	and define $p$ according to Equation~\eqref{def:p value using permutations}. Suppose further
	 \begin{align*}
	 	\mmdulower{\bx}{\by}\le \mmdu{\bx}{\by},~ \mmdu{\bx_{\sigma^{(\iconstant)}}}{\by_{\sigma^{(\iconstant)}}} \le \mmduupper{\bx_{\sigma^{(\iconstant)}}}{\by_{\sigma^{(\iconstant)}}},~\iconstant = 1,\ldots, B.
	 \end{align*}
	 Define
	 \begin{align*}
	 	\upperp = \frac{1}{B + 1} \left(1 + \sum_{\iconstant = 1}^{B} \indicator{\mmduupper{\bx_{\sigma^{(\iconstant)}}}{\by_{\sigma^{(\iconstant)}}} \ge \mmdulower{\bx}{\by}} \right).
	 \end{align*}
	 Then, we must have $\upperp \ge \p$.
\end{theorem}
The proof of this theorem is included in Appendix \ref{Proof of Theorem 3}. 
In the presence of missing data, $\mmdu{\bx}{\by}$ cannot be computed 
directly. However, as shown in Section \ref{Bounding MMD with Missing Data}, 
the lower and upper bounds of $\mmdu{\bx}{\by}$ can be 
computed when the Laplacian kernel is used. 
Then, Theorem \ref{theorem:3} can be applied to compute 
$\upperp$, a value that is bounded by $\p$, defined 
in Equation~\eqref{def:p value using permutations}. 
It is proved in \cite{schrab2023mmd} that, under the null hypothesis,  
computing the $p$-value using permutations controls the Type I error.
Therefore, using the approach described in Theorem~\ref{theorem:3}
will also control the Type I error,  
since the bounds were derived without any 
assumptions about the missing data mechanisms.

\paragraph{Bounding $p$-value using normality approximation.} 
In the above, we have proposed a testing procedure based on permutation 
test in the presence of missing data. Notably, this procedure provides a 
valid test procedure for any given $\nln{1},\nln{2},\dd$.
When both $\nln{1},\nln{2},\dd$ are large enough, an alternative method 
to compute a $p$-value is to use 
the normality approximation \cite{gao2023two}. 
It is suggested in \cite{gao2023two} that this approximation 
is effective when $\nln{1},\nln{2} \ge 25, \dd \ge 50$. 

To more formally describe the normality approximation proposed in \cite{gao2023two}, we will follow their notation in this section. Let us denote the sample sizes for $X$ and $Y$ as $\n,\m$, respectively, rather than $\nln{1},\nln{2}$. It is proved in \cite{gao2023two} the studentized  $\mmdu{\bx}{\by}$, taking the form as
\begin{align} \label{def:studentized mmd}
	\frac{\mmdu{\bx}{\by}}{\sqrt{\constant{\n,\m}{\mathcal{V}}_{n,m}^{k*}}},
\end{align}
where $\constant{\n,\m} = 2/(n(n-1)) + 4/(nm) + 2/(m(m-1))$, and $\mathcal{V}_{n,m}^{k*}$ is defined in Proposition 10 in \cite{gao2023two}, convergence to standard normal distribution under the null (see Theorem 16 in \cite{gao2023two}). In order to use this result, it is necessary to compute the estimation of variance of $\mmdu{\bx}{\by}$, i.e. $\constant{\n,\m}\mathcal{V}_{n,m}^{k*}$, which can not be computed directly with missing data. To overcome this problem, we propose to first bound $\mathcal{V}_{n,m}^{k*}$, leading to the following theorem: 

\begin{theorem} \label{theorem:4}
    A lower bound for
    the studentized MMD statistic in Equation \ref{def:studentized mmd} can be computed, and this 
    provides an upper bound for the $p$-value computed via the normality approximation, 
    when not all data are fully missing.
\end{theorem}

This theorem is proved in Appendix \ref{Proof of Theorem 4}, and provides an
approach for using the normality approximation to compute the $p$-value, when 
data are missing.

\section{Experiments} \label{Experiments}

In this section, we investigate the Type I error and statistical 
power of MMD-Miss proposed in Section \ref{Testing with Missing Data Using Bounds of MMD} 
and compare it with three common missing data approaches:
case deletion, mean imputation and hot deck imputation. 
These other three approaches are described in further detail in 
Appendix~\ref{simulationsettings}. 


\paragraph{Asymptotic Type I error for a small proportion of missing data (5\%).}

The first experiment investigates the Type I error of the proposed methods 
and common missing data approaches, with increasing sample sizes for a given 
proportion (5\%) of missing data, where data are missing not at random (MNAR). 
To assess the Type I error,
random samples $\bx = \{\xln{1}, \ldots, \xln{\nln{1}}\}$, 
$\by = \{\yln{1}, \ldots, \yln{\nln{2}}\}$ are independently generated
from a $d$-dimensional normal distribution with mean vector 
$\mu = (0, \ldots, 0)^{\bt}$ and covariance matrix equal to the identity matrix.
Subsequently, $5\%$ of samples 
will be selected in both $\bx$ and $\by$ to be labeled as missing 
(for $d = 1$), or incomplete (for $d > 1$). When an observation is labeled as 
incomplete, $30\%$ of its components will be missing values.

In this experiment, $d \in \{1, 10, 50\}$ will be considered. 
When the dimension $d = 1$, only observations $\xln{\iconstant} \in \bx$
with $\xln{\iconstant} < 0$ will be randomly selected to possibly be missing, 
while only observations $\yln{\iconstant} \in \by$ with $\yln{\iconstant} > 0$ will 
be randomly selected to possibly be missing; in this case
the data are missing not at random (MNAR), i.e. the data 
are informatively missing.
For higher dimensions $d \in \{ 10, 50 \}$, the observations 
in $\bx$ that will possibly be partially missing are those with 
$\sum_{\lconstant = 1}^{\dd} \xln{\iconstant}(\lconstant)/\sqrt{d} < -0.8$, and 
then for each chosen sample, $30\%$ of its 
components with values smaller than 
$\text{median}{(\xln{\iconstant}(1),\ldots,\xln{\iconstant}(\dd))}$ 
will be randomly selected to be missing. On the other hand, only observations 
in $\by$ with 
$\sum_{\lconstant = 1}^{\dd} \yln{\iconstant}(\lconstant)/\sqrt{d} > 0.8$ 
will possibly be partially missing, 
and again for each chosen observation, $30\%$ its components 
with values larger than 
$\text{median}{(\yln{\iconstant}(1),\ldots,\yln{\iconstant}(\dd))}$ 
will be randomly selected to be missing. We use these missingness mechanisms, 
since if only the fully observed observations are taken into account, 
the two samples will appear to be different.

Since using the normality approximation with MMD-Miss may not be suitable 
when $\dd \in \{1, 10\}$, we only assess its performance when $\dd = 50$, 
the dimension suggested by \cite{gao2023two} for when the approximation
will be effective. Figure~\ref{fig:1} shows that the Type I error of all 
common missing data approaches increases with increasing sample sizes. 
In particular, when the sample sizes equals to $5000$, the Type I error 
approaches to 1. In contrast, MMD-Miss controls the Type I 
error. The significance level is  $\alpha = 0.05$.

\begin{figure}[h] 
	\begin{center} 
		\includegraphics[width=14.1cm]{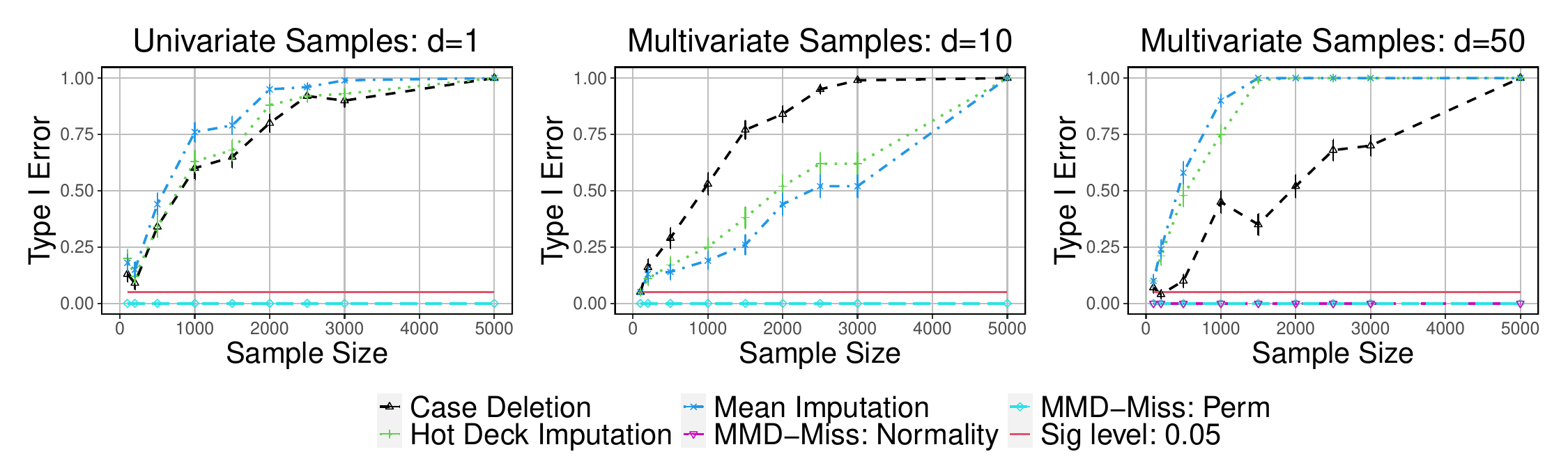}
	\end{center}
	\caption{The Type I error of the proposed method and common missing data methods
		when data are missing not at random. 
		(Left): $d = 1$; (Middle): $d = 10$; (Right): $d = 50$. 
		MMD-Miss only uses the normality approximation when $d = 50$. 
		For all figures, the significance level is $\alpha = 0.05$, 
		and $5\%$ of the data are missing or partially observed.
		The plotted values show the average times of 
		the null hypothesis is rejected over 100 repetitions. The error bars represent one standard error of the mean.}  
	\label{fig:1}
\end{figure}


\paragraph{Type I error and power for univariate data.}

The second experiment considers the Type I error and power of MMD-Miss and 
common missing data approaches for univariate samples where data are 
missing not at random (MNAR). 
To assess the Type I error, the data $\bx = \{\xln{1}, \ldots, \xln{\nln{1}}\}$ 
and $\by = \{\yln{1}, \ldots, \yln{\nln{2}}\}$ are independently sampled 
from a standard normal distribution $N(0,1)$. To assess the statistical power, 
the $\bx$ data are sampled independently from $N(0,1)$ 
the $\by$ data are sampled independently from either (a) $N(1,1)$ or (b) $N(1.5,1)$. 
Then, a proportion $s \in \{0,0.01,\ldots,0.20\}$ of samples are selected 
in $\bx$ and $\by$ to be missing. The missingness mechanisms are
the same as for Figure \ref{fig:1} when $\dd = 1$; 
only observations $\xln{\iconstant} \in \bx$
with $\xln{\iconstant} < 0$ and observations $\yln{\iconstant} \in \by$ 
with $\yln{\iconstant} > 0$ will 
be randomly selected to possibly be missing.

\begin{figure}[h] 
	\begin{center} 
		\includegraphics[width=14.1cm]{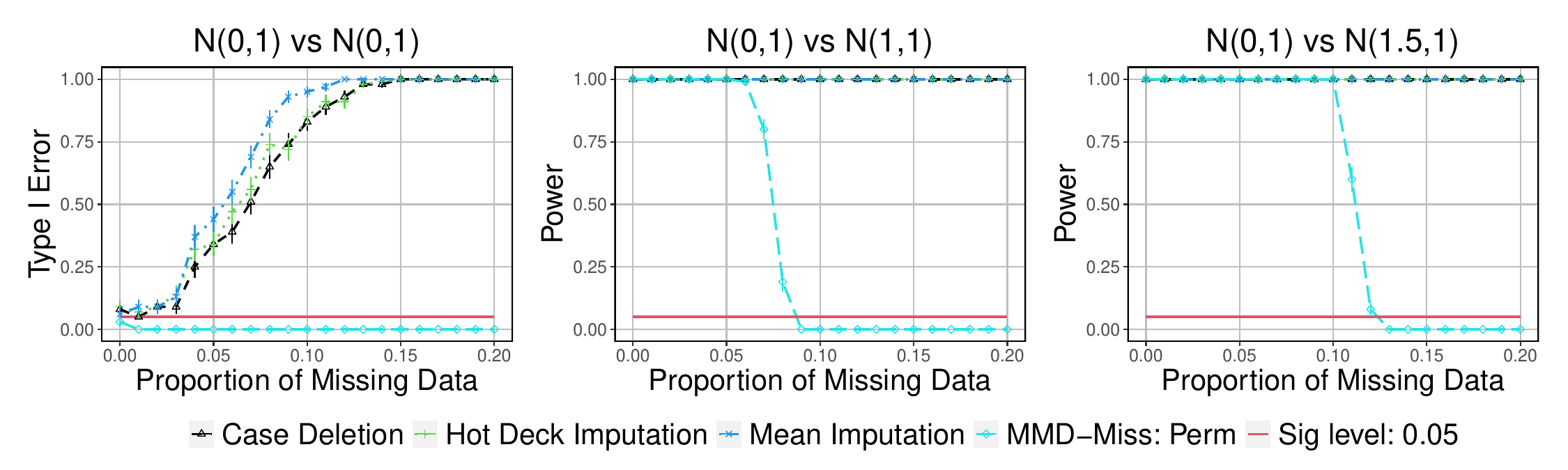}
	\end{center}
	\caption{The Type I error and power of the proposed method and common missing data methods
		for univariate samples when data are missing not at random. 
		(Left): $N(0,1)$ vs $N(0,1)$; 
		(Middle): $N(0,1)$ vs $N(1,1)$; 
		(Right): $N(0,1)$ vs $N(1.5,1)$. 
		A significance level $\alpha = 0.05$ and sample sizes $\nln{1} = \nln{2} = 500$ 
		are used. 
		The plotted values show the average times of 
		the null hypothesis is rejected over 100 repetitions. The error bars represent one standard error of the mean.}  
	\label{fig:2}
\end{figure}

Figure \ref{fig:2} shows that the Type I error is not controlled by 
the case deletion, hot deck imputation or mean imputation methods. 
When the proportion of missing data is $s=5\%$, the Type I error 
is above $25\%$ for each of these three methods, and when the 
$s=10\%$ the Type I error of these three methods is beyond $75\%$. 
On the other hand, the proposed method MMD-Miss controls the 
Type I error, regardless of the proportion of missing data. 
While all the common missing data methods have good statistical 
power, the power of the proposed method decreases significantly when 
there are more than $5\%$ missing data for alternative 
$N(0,1)$ vs $N(1,1)$, and when there are more than $10\%$ missing 
data for alternative $N(0,1)$ vs $N(1.5,1)$. 
This experiment demonstrates that MMD-Miss is useful when the 
proportion of missing data is in the range $5\%$ to $10\%$, 
while the three common missing data approaches fail to 
control the Type I error.

\paragraph{Type I error and power for multivariate data.} 
The third experiment compares the Type I error and power of 
MMD-Miss and the three common missing data approaches for 
multivariate observations.
For a value $a \in \mathbb{R}$, let 
$\mu_{a, \dd} = (a, a, \dots, a)^{\bt}$ be a 
mean vector with $\dd$ components all equal to $a$.
To assess the Type I error, 
observations $\bx = \{\xln{1}, \ldots, \xln{\nln{1}}\}$ and 
$\by = \{\yln{1}, \ldots, \yln{\nln{2}}\}$ are independently sampled 
from the $d$-dimensional normal distribution with zero mean vector 
$\mu_{0, \dd} = (0, 0, \dots, 0)^{\bt}$ and covariance matrix 
$\Sigma = I_{\dd}$, the identity matrix. 
To assess the statistical power, 
the $\bx$ data are independently sampled
from $N( \mu_{0, \dd}, I_{\dd})$ 
and either (a) $N( \mu_{1, \dd}, I_{\dd})$ or (b) $N(\mu_{1.5, \dd}, I_{\dd})$. 
Then, a proportion $s \in [0,0.2]$ of observations from 
both $\bx$ and $\by$ will be randomly selected to 
be partially missing.  
The missingness mechanisms follows the same approach used 
for Figure \ref{fig:1} when $\dd > 1$, with each partially missing
observation having $30\%$ missing components. 
The results displayed in Figure \ref{fig:3} show a similar pattern to the results for
univariate case: (i) the three common missing data approaches cannot control 
the Type I error 
although they all demonstrate good power, 
while (ii) MMD-Miss controls the Type I error, while also have good power 
when proportion of missing data is around $5\%$ to $10\%$. 
Notably, in the case where $\dd=50$, MMD-Miss with the normality 
approximation has better statistical power 
than when it uses permutations to compute the $p$-value. 
For the alternative $N((1.5, \ldots, 1.5)^{\bt}, I_{\dd})$, 
MMD-Miss with the normality approximation has good power when $15\%$ of the data 
have missing values.

\begin{figure}[h] 
	\begin{center} 
		\includegraphics[width=14.1cm]{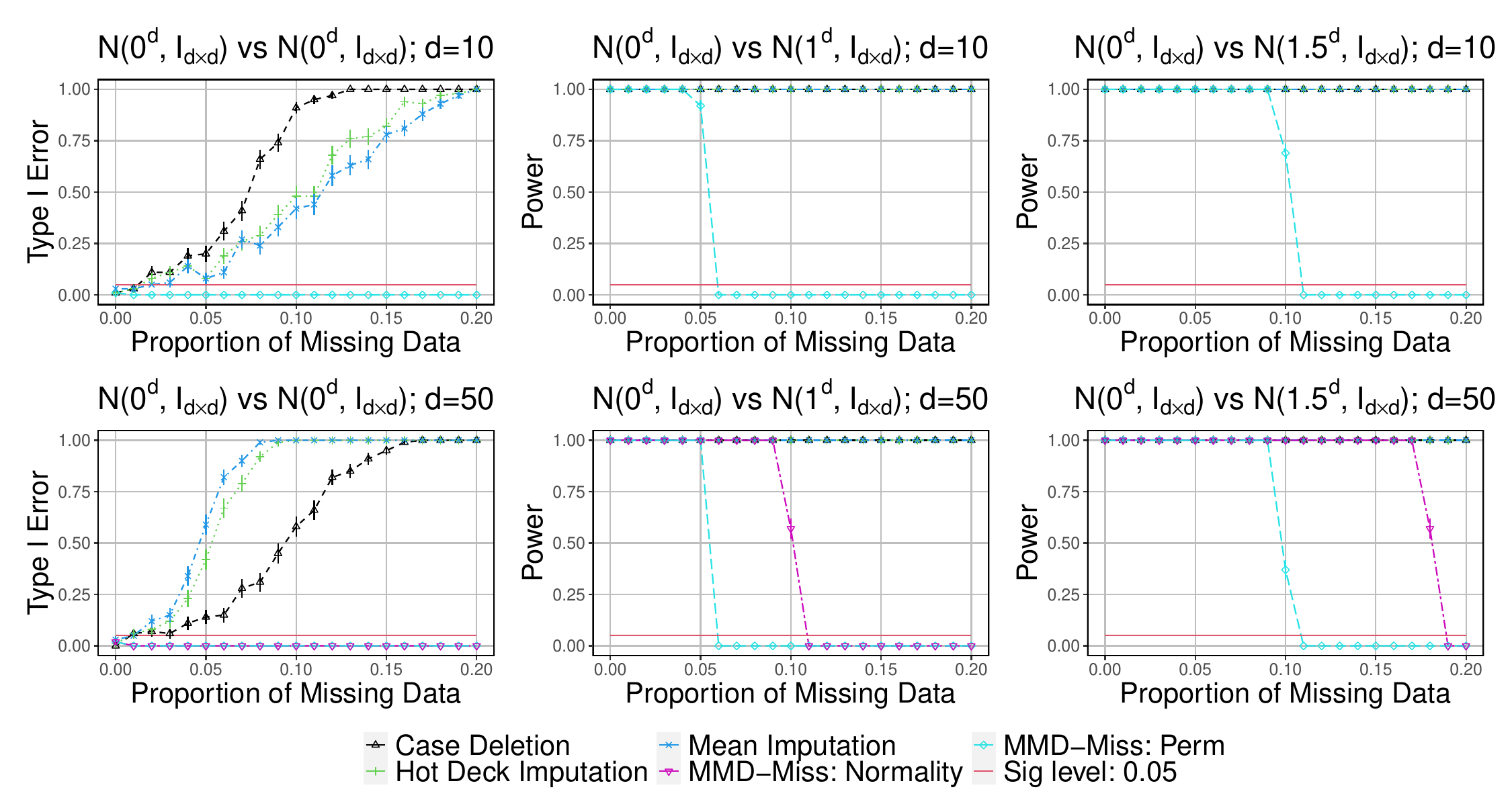}
	\end{center}
	\caption{The Type I error and power of MMD-Miss and the three common 
		missing data approaches for multivariate samples when data are missing not at random. 
		when the dimension of the data is either  $\dd = 10$ or $\dd = 50$. 
		When $\dd = 50$, the normality approximation can be used with MMD-Miss. 
		A significance level $\alpha = 0.05$, and sample sizes $\nln{1} = \nln{2} = 500$ 
		are used. 
		The plotted values show the average times of 
		the null hypothesis is rejected over 100 repetitions. The error bars represent one standard error of the mean.}  
	\label{fig:3}
\end{figure}

\paragraph{MNIST dataset.} We evaluate the performance of MMD-Miss on 
real-world data using MNIST images \cite{lecun1998gradient}, with examples 
shown in Figure~\ref{fig:MNIST} in Appendix~\ref{Examples of MNIST dataset}. 
Each image in the MNIST dataset has dimensions of $28 \times 28$ pixels 
and is labeled from $0$ to $9$. For our analysis, the pixel values of 
each image are scaled between $0$ and $1$. To assess Type I error, 
datasets $X$ and $Y$ are generated by randomly sampling with replacement 
from MNIST training set images labeled as $3$. For evaluating power, 
$X$ is generated by randomly sampling with replacement from images labeled as 
$0$, while $Y$ continues to be sampled from images labeled as $3$. 
A proportion $s \in [0,0.25]$ of the samples in $Y$ are then 
randomly selected and labeled as incomplete if there are more 
than 85 non-zero pixels in the region defined by rows 1 to 14 
and columns 8 to 21 (i.e. a sub-block in the upper half of each image), 
which will be marked as missing. In other words, 
images with more non-zero pixels in the specified region are more likely 
to be partially observed. Examples of these incompletely observed 
images are shown in Figure \ref{fig:MNIST Incomplete} in 
Appendix~\ref{Examples of MNIST dataset}
Table \ref{tab:1} shows that the three 
common missing data approaches cannot control the Type I error 
for this task. On the other hand, MMD-Miss controls the Type I error while enjoying good power. For MMD-Miss based on normality approximation, the power is good 
except for the case when $25\%$ images from $\by$ are missing. 

\begin{table}[h]
	\caption{Comparison of Type I error and power on MNIST dataset. 
		NA stands for MMD-Miss with the normality approximation; 
		PT for MMD-Miss with a permutation test; 
		CD stands for case deletion; 
		MI stands for mean imputation; 
		HD stands for hot deck imputation. 
		A significance level $\alpha = 0.05$, 
		and sample sizes $\nln{1} = \nln{2} = 500$ are used. 
		The values in the table are the average times 
		the null hypothesis is rejected over 100 repetitions.}
	\label{tab:1}
	\centering
	\begin{tabular}{cccccccccccc}
		\toprule
		\multirow{4}{*}{\begin{tabular}[c]{@{}c@{}}Proportion \\of \\ Missing \\ Data ($s$) \end{tabular}} & \multicolumn{5}{c}{Type I Error} & \multicolumn{5}{c}{Power} \\
		\cmidrule(r){2-6} \cmidrule(r){7-11}
		& \multicolumn{2}{c}{MMD-Miss} & \multicolumn{3}{c}{Common} & \multicolumn{2}{c}{MMD-Miss} & \multicolumn{3}{c}{Common} \\
		\cmidrule(r){2-3} \cmidrule(r){4-6} \cmidrule(r){7-8} \cmidrule(r){9-11}
		& NA & PT & CD & MI & HD & NA & PT & CD & MI & HD \\
		\midrule
		0.00 & 0.07 & 0.07 & 0.06 & 0.07 & 0.07 & 1.00 & 1.00 & 1.00 & 1.00 & 1.00 \\
		0.05 & 0.00 & 0.00 & {0.08} &{1.00} & {1.00} & 1.00 & 1.00 & 1.00 & 1.00 & 1.00\\
		0.10 & 0.00 & 0.00 & {0.21} & {1.00} & {1.00} & 1.00 & 1.00 & 1.00 & 1.00 & 1.00\\
		0.15 & 0.00 & 0.00 & {0.64} & {1.00} & {1.00} & 1.00 & 0.00 & 1.00 & 1.00 & 1.00\\
		0.20 & 0.00 & 0.00 & {0.91} & {1.00} & {1.00} & 0.96 & 0.00 & 1.00 & 1.00 & 1.00\\
		0.25 & 0.00 & 0.00 & {1.00} & {1.00} & {1.00} & 0.00 & 0.00 & 1.00 & 1.00 & 1.00\\
		\bottomrule
	\end{tabular}
\end{table}

\section{Conclusion} \label{Conclusion}

The proposed method MMD-Miss can perform two-sample testing on both
univariate and multivariate data with missing values.
It has good statistical power, typically when $5\%$ to $10\%$ of the data 
are missing, while controlling the Type I error. 

\paragraph{Limitations} First, a limitation of MMD-Miss is that 
it is restricted to the Laplacian kernel for computational reasons.
However since the Laplacian kernel is characteristic, MMD with 
this kernel can detect any distributional change. Moreover, 
MMD-Miss will be effective for other kernels if appropriate
bounds can be derived and computed efficiently.
Second, MMD-Miss is only suitable when up
to $20\%$ of the values are not fully observed, and typically performs best
when $5\%$ to $10\%$ of the data are missing. However, using
imputation methods when even $5\%$ of the data are informatively missing
leads to the Type I error being out of control.

\section{Acknowledgements} \label{sec:acknowledgements}

Yijin Zeng is funded by a Roth Studentship from the 
Department of Mathematics, Imperial College London
and the 
EPSRC CDT in Statistics and Machine Learning.


\clearpage
\appendix

\section*{Overview of Appendix}

Appendix \ref{mathematicaldetails}. Mathematical details of developing bounds of $\mmdu{\bx}{\by}$ in Section \ref{Bounding MMD with Missing Data} and constructing valid two-sample testing method in Section \ref{Testing with Missing Data Using Bounds of MMD} are given. 

Appendix \ref{simulationdetails}. Details of experiments in Section \ref{Experiments} are provided.  

Appendix \ref{addtionalexperiments}. Additional experiments are given for further investigating the power of the proposed methods.

\section{Mathematical Details} \label{mathematicaldetails}

\subsection{Decomposition of MMD} \label{supp:decomposition}


Without loss of generality, let us 
assume $\xln{1}, \ldots, \xln{\mln{1}}$, $\yln{1}, \ldots, \yln{\mln{2}}$ are 
samples that are not observed, or not fully observed. Then, $\mmdu{\bx}{\by}$ can be decomposed, using the following lemma: 

\begin{lemma} \label{lemma:1}
	Suppose 
    $\bx = \{\xln{1}, \ldots, \xln{\nln{1}} \}$, 
    $\by = \{\yln{1}, \ldots, \yln{\nln{2}}\}$ are subsets of 
    $\mathbb{R}^d$, where $\dd \ge 1$. 
    Suppose that $\mmdu{\bx}{\by}$ is unbiased MMD test statistic defined as
    \begin{align} 
    	\begin{split}
    		\mmdu{\bx}{\by} &= \tfrac{1}{\nln{1}(\nln{1}-1)} \sum_{\iconstant = 1}^{\nln{1}} \sum_{ \substack{\jconstant = 1 \\ \jconstant \neq \iconstant } }^{\nln{2}} \kernel{\xln{\iconstant}}{\xln{\jconstant}}
    		+ \tfrac{1}{\nln{2}(\nln{2}-1)} \sum_{\iconstant = 1}^{\nln{2}} \sum_{ \substack{\jconstant = 1 \\ \jconstant \neq \iconstant } }^{\nln{2}} \kernel{\yln{\iconstant}}{\yln{\jconstant}}
    		- \tfrac{2}{\nln{1}\nln{2}}	\sum_{\iconstant = 1}^{\nln{1}} \sum_{\jconstant = 1}^{\nln{2}} \kernel{\xln{\iconstant}}{\yln{\jconstant}},
    		\nonumber
    	\end{split}
    \end{align}
    with $\kk$ denoting the Laplacian kernel. Then, for any two positive integers $\mln{1}$, $\mln{2}$ such that $\mln{1}~\le~\nln{1},~\mln{2}~\le~\nln{2}$, $\mmdu{\bx}{\by}$ can be rewritten as:
	\begin{align*}
		& \mmdu{\bx}{\by} = \termone + \termtwo + \termthree + \termfour,
	\end{align*}
	where
    $\constant{1} = \frac{2}{\nln{1}(\nln{1}-1)}$, 
    $\constant{2} = \frac{2}{\nln{2}(\nln{2}-1)}$, 
    $\constant{3} = \frac{2}{\nln{1}\nln{2}}$, and
	\begin{align*}
		\termone &= \constant{1} \sum_{\iconstant = 1}^{\mln{1}} \sum_{\jconstant = \iconstant + 1}^{\mln{1}} \kernel{\xln{\iconstant}}{\xln{\jconstant}} + \constant{2} \sum_{\iconstant = 1}^{\mln{2}} \sum_{\jconstant = \iconstant +1}^{\mln{2}} \kernel{\yln{\iconstant}}{\yln{\jconstant}} - \constant{3} \sum_{\iconstant = 1}^{\mln{1}} \sum_{\jconstant = 1}^{\mln{2}} \kernel{\xln{\iconstant}}{\yln{\jconstant}},\\
		\termtwo &= \constant{1} \sum_{\iconstant = \mln{1} + 1}^{\nln{1}-1} \sum_{\jconstant = \iconstant + 1}^{\nln{1}} \kernel{\xln{\iconstant}}{\xln{\jconstant}} + \constant{2}\sum_{\iconstant = \mln{2}+1}^{\nln{2}-1} \sum_{\jconstant = \iconstant +1}^{\nln{2}} \kernel{\yln{\iconstant}}{\yln{\jconstant}} - 	\constant{3}	\sum_{\iconstant = \mln{1} + 1}^{\nln{1}} \sum_{\jconstant = \mln{2} + 1}^{\nln{2}} \kernel{\xln{\iconstant}}{\yln{\jconstant}},\\
		\termthree & =  \constant{1}\sum_{\iconstant = 1}^{\mln{1}} \sum_{\jconstant = \mln{1} + 1}^{\nln{1}} \kernel{\xln{\iconstant}}{\xln{\jconstant}} - 	\constant{3} \sum_{\iconstant = 1}^{\mln{1}} \sum_{\jconstant = \mln{2} + 1}^{\nln{2}} \kernel{\xln{\iconstant}}{\yln{\jconstant}}, \\
		\termfour & = \constant{2} \sum_{\iconstant = 1}^{\mln{2}} \sum_{\jconstant = \mln{2} +1}^{\nln{2}} \kernel{\yln{\iconstant}}{\yln{\jconstant}} - 	\constant{3} \sum_{\iconstant = \mln{1} + 1}^{\nln{1}} \sum_{\jconstant = 1}^{\mln{2}} \kernel{\xln{\iconstant}}{\yln{\jconstant}}.
	\end{align*}
\end{lemma}

\begin{proof}
	To start,  for any $\x$, $\y$ of $\mathbb{R}^{\dd}$ real values, the Laplacian kernel is defined as follows
	\begin{align*}
		\kk({\x},{\y}) = \exp(-\beta ||\x - \y||_1).
	\end{align*}
	It can be seen that $\kk$ is symmetric, i.e. $\kk(\x,\y) = \kk(\y,\x)$. Subsequently, according to Lemma 2 in \cite{bodenham2023eummd}, the $\mmdu{\bx}{\by}$ can then be rewritten as
	\begin{align*}
		\mmdu{\bx}{\by} &:= \frac{2}{\nln{1}(\nln{1}-1)} \sum_{\iconstant = 1}^{\nln{1}} \sum_{\jconstant = 1}^{\iconstant -1 } \kernel{\xln{\iconstant}}{\xln{\jconstant}}
		+ \frac{2}{\nln{2}(\nln{2}-1)} \sum_{\iconstant = 1}^{\nln{2}} \sum_{\jconstant = 1}^{\iconstant - 1 } \kernel{\yln{\iconstant}}{\yln{\jconstant}} \\
		& - \frac{2}{\nln{1}\nln{2}}	\sum_{\iconstant = 1}^{\nln{1}} \sum_{\jconstant = 1}^{\nln{2}} \kernel{\xln{\iconstant}}{\yln{\jconstant}}.
	\end{align*}
	Notice that
	\begin{align*}
		\sum_{\iconstant = 1}^{\nln{1}} \sum_{\jconstant = 1}^{\iconstant -1 } \kernel{\xln{\iconstant}}{\xln{\jconstant}} &= \sum_{\iconstant = 1}^{\nln{1}-1} \sum_{\jconstant = \iconstant + 1}^{\nln{1}} \kernel{\xln{\iconstant}}{\xln{\jconstant}},\\
		\sum_{\iconstant = 1}^{\nln{2}} \sum_{\jconstant = 1}^{\iconstant - 1 } \kernel{\yln{\iconstant}}{\yln{\jconstant}} & = \sum_{\iconstant = 1}^{\nln{2}-1} \sum_{\jconstant = \iconstant +1}^{\nln{2}} \kernel{\yln{\iconstant}}{\yln{\jconstant}}.
	\end{align*}
	Thus, we further have
	\begin{align*}
		\mmdu{\bx}{\by} &:= \frac{2}{\nln{1}(\nln{1}-1)} \sum_{\iconstant = 1}^{\nln{1}-1} \sum_{\jconstant = \iconstant + 1}^{\nln{1}} \kernel{\xln{\iconstant}}{\xln{\jconstant}}
		+ \frac{2}{\nln{2}(\nln{2}-1)} \sum_{\iconstant = 1}^{\nln{2}-1} \sum_{\jconstant = \iconstant +1}^{\nln{2}} \kernel{\yln{\iconstant}}{\yln{\jconstant}} \\
		& - \frac{2}{\nln{1}\nln{2}}	\sum_{\iconstant = 1}^{\nln{1}} \sum_{\jconstant = 1}^{\nln{2}} \kernel{\xln{\iconstant}}{\yln{\jconstant}}.
	\end{align*}
	Let $\mln{1}, \mln{2}$ be any two positive integers such that $\mln{1} \le \nln{1}, \mln{2} \le \nln{2}$. Denote $\constant{1} = \frac{2}{\nln{1}(\nln{1}-1)}$. Subsequently,
	\begin{align*}
		\frac{2}{\nln{1}(\nln{1}-1)}\sum_{\iconstant = 1}^{\nln{1}-1} \sum_{\jconstant = \iconstant + 1}^{\nln{1}} \kernel{\xln{\iconstant}}{\xln{\jconstant}} &= \constant{1} \sum_{\iconstant = 1}^{\mln{1}} \sum_{\jconstant = \iconstant + 1}^{\nln{1}} \kernel{\xln{\iconstant}}{\xln{\jconstant}} + \constant{1} \sum_{\iconstant = \mln{1} + 1}^{\nln{1}-1} \sum_{\jconstant = \iconstant + 1}^{\nln{1}} \kernel{\xln{\iconstant}}{\xln{\jconstant}}\\
		& = \constant{1}\sum_{\iconstant = 1}^{\mln{1}} \sum_{\jconstant = \iconstant + 1}^{\mln{1}} \kernel{\xln{\iconstant}}{\xln{\jconstant}} + \constant{1} \sum_{\iconstant = 1}^{\mln{1}} \sum_{\jconstant = \mln{1} + 1}^{\nln{1}} \kernel{\xln{\iconstant}}{\xln{\jconstant}}\\
		& + \constant{1}\sum_{\iconstant = \mln{1} + 1}^{\nln{1}-1} \sum_{\jconstant = \iconstant + 1}^{\nln{1}} \kernel{\xln{\iconstant}}{\xln{\jconstant}}.
	\end{align*}
	Also, denote $\constant{2} = \frac{2}{\nln{2}(\nln{2}-1)}$, it follows
	\begin{align*}
		\frac{2}{\nln{2}(\nln{2}-1)} \sum_{\iconstant = 1}^{\nln{2}-1} \sum_{\jconstant = \iconstant +1}^{\nln{2}} \kernel{\yln{\iconstant}}{\yln{\jconstant}} &= \constant{2} \sum_{\iconstant = 1}^{\mln{2}} \sum_{\jconstant = \iconstant +1}^{\nln{2}} \kernel{\yln{\iconstant}}{\yln{\jconstant}} + \constant{2} \sum_{\iconstant = \mln{2}+1}^{\nln{2}-1} \sum_{\jconstant = \iconstant +1}^{\nln{2}} \kernel{\yln{\iconstant}}{\yln{\jconstant}} \\
		& = \constant{2} \sum_{\iconstant = 1}^{\mln{2}} \sum_{\jconstant = \iconstant +1}^{\mln{2}} \kernel{\yln{\iconstant}}{\yln{\jconstant}} +  \constant{2} \sum_{\iconstant = 1}^{\mln{2}} \sum_{\jconstant = \mln{2} +1}^{\nln{2}} \kernel{\yln{\iconstant}}{\yln{\jconstant}} \\
		& +\constant{2}\sum_{\iconstant = \mln{2}+1}^{\nln{2}-1} \sum_{\jconstant = \iconstant +1}^{\nln{2}} \kernel{\yln{\iconstant}}{\yln{\jconstant}}.
	\end{align*}
	Denote $\constant{3} = \frac{2}{\nln{1}\nln{2}}$, then
	\begin{align*}
		\frac{2}{\nln{1}\nln{2}}	\sum_{\iconstant = 1}^{\nln{1}} \sum_{\jconstant = 1}^{\nln{2}} \kernel{\xln{\iconstant}}{\yln{\jconstant}} &= \constant{3}	\sum_{\iconstant = 1}^{\mln{1}} \sum_{\jconstant = 1}^{\nln{2}} \kernel{\xln{\iconstant}}{\yln{\jconstant}} + \constant{3}	\sum_{\iconstant = \mln{1} + 1}^{\nln{1}} \sum_{\jconstant = 1}^{\nln{2}} \kernel{\xln{\iconstant}}{\yln{\jconstant}}\\
		& = \constant{3}	\sum_{\iconstant = 1}^{\mln{1}} \sum_{\jconstant = 1}^{\mln{2}} \kernel{\xln{\iconstant}}{\yln{\jconstant}}  + \constant{3}\sum_{\iconstant = 1}^{\mln{1}} \sum_{\jconstant = \mln{2} + 1}^{\nln{2}} \kernel{\xln{\iconstant}}{\yln{\jconstant}} \\
		& + \constant{3} \sum_{\iconstant = \mln{1} + 1}^{\nln{1}} \sum_{\jconstant = 1}^{\mln{2}} \kernel{\xln{\iconstant}}{\yln{\jconstant}} + \constant{3}	\sum_{\iconstant = \mln{1} + 1}^{\nln{1}} \sum_{\jconstant = \mln{2} + 1}^{\nln{2}} \kernel{\xln{\iconstant}}{\yln{\jconstant}}.
	\end{align*}
	Combine the above together, $\mmdu{\bx}{\by}$ can be rewritten as
	\begin{align*}
		\mmdu{\bx}{\by} &= \constant{1} \sum_{\iconstant = 1}^{\mln{1}} \sum_{\jconstant = \iconstant + 1}^{\mln{1}} \kernel{\xln{\iconstant}}{\xln{\jconstant}} + \constant{1}  \sum_{\iconstant = 1}^{\mln{1}} \sum_{\jconstant = \mln{1} + 1}^{\nln{1}} \kernel{\xln{\iconstant}}{\xln{\jconstant}}\\
		& + \constant{1}  \sum_{\iconstant = \mln{1} + 1}^{\nln{1}-1} \sum_{\jconstant = \iconstant + 1}^{\nln{1}} \kernel{\xln{\iconstant}}{\xln{\jconstant}} + \constant{2}  \sum_{\iconstant = 1}^{\mln{2}} \sum_{\jconstant = \iconstant +1}^{\mln{2}} \kernel{\yln{\iconstant}}{\yln{\jconstant}} \\
		&+  \constant{2}  \sum_{\iconstant = 1}^{\mln{2}} \sum_{\jconstant = \mln{2} +1}^{\nln{2}} \kernel{\yln{\iconstant}}{\yln{\jconstant}}
		+ \constant{2}  \sum_{\iconstant = \mln{2}+1}^{\nln{2}-1} \sum_{\jconstant = \iconstant +1}^{\nln{2}} \kernel{\yln{\iconstant}}{\yln{\jconstant}}\\
		&- \constant{3} 	\sum_{\iconstant = 1}^{\mln{1}} \sum_{\jconstant = 1}^{\mln{2}} \kernel{\xln{\iconstant}}{\yln{\jconstant}}  - \constant{3}  \sum_{\iconstant = 1}^{\mln{1}} \sum_{\jconstant = \mln{2} + 1}^{\nln{2}} \kernel{\xln{\iconstant}}{\yln{\jconstant}} \\
		& - \constant{3}  \sum_{\iconstant = \mln{1} + 1}^{\nln{1}} \sum_{\jconstant = 1}^{\mln{2}} \kernel{\xln{\iconstant}}{\yln{\jconstant}} - \constant{3} 	\sum_{\iconstant = \mln{1} + 1}^{\nln{1}} \sum_{\jconstant = \mln{2} + 1}^{\nln{2}} \kernel{\xln{\iconstant}}{\yln{\jconstant}}.
	\end{align*}
	By rearranging the above equation, we conclude
	\begin{align*}
		\mmdu{\bx}{\by} &= \constant{1} \sum_{\iconstant = 1}^{\mln{1}} \sum_{\jconstant = \iconstant + 1}^{\mln{1}} \kernel{\xln{\iconstant}}{\xln{\jconstant}} + \constant{2}  \sum_{\iconstant = 1}^{\mln{2}} \sum_{\jconstant = \iconstant +1}^{\mln{2}} \kernel{\yln{\iconstant}}{\yln{\jconstant}} - \constant{3} 	\sum_{\iconstant = 1}^{\mln{1}} \sum_{\jconstant = 1}^{\mln{2}} \kernel{\xln{\iconstant}}{\yln{\jconstant}} \\
		& + \constant{1} \sum_{\iconstant = \mln{1} + 1}^{\nln{1}-1} \sum_{\jconstant = \iconstant + 1}^{\nln{1}} \kernel{\xln{\iconstant}}{\xln{\jconstant}} + \constant{2}\sum_{\iconstant = \mln{2}+1}^{\nln{2}-1} \sum_{\jconstant = \iconstant +1}^{\nln{2}} \kernel{\yln{\iconstant}}{\yln{\jconstant}} - 	\constant{3}	\sum_{\iconstant = \mln{1} + 1}^{\nln{1}} \sum_{\jconstant = \mln{2} + 1}^{\nln{2}} \kernel{\xln{\iconstant}}{\yln{\jconstant}} \\
		&+  \constant{1}\sum_{\iconstant = 1}^{\mln{1}} \sum_{\jconstant = \mln{1} + 1}^{\nln{1}} \kernel{\xln{\iconstant}}{\xln{\jconstant}} - 	\constant{3} \sum_{\iconstant = 1}^{\mln{1}} \sum_{\jconstant = \mln{2} + 1}^{\nln{2}} \kernel{\xln{\iconstant}}{\yln{\jconstant}}\\
		&- \constant{2} \sum_{\iconstant = 1}^{\mln{2}} \sum_{\jconstant = \mln{2} +1}^{\nln{2}} \kernel{\yln{\iconstant}}{\yln{\jconstant}} - 	\constant{3} \sum_{\iconstant = \mln{1} + 1}^{\nln{1}} \sum_{\jconstant = 1}^{\mln{2}} \kernel{\xln{\iconstant}}{\yln{\jconstant}},
	\end{align*}
	which proves our result.
\end{proof}

This lemma divides the $\mmdu{\bx}{\by}$ into four parts: the first part $\termone$ includes the unobserved samples in $\bx$ and $\by$ only; on the contrary, the second part $\termtwo$ includes observed samples in $\bx$ and $\by$ only; the third and the fourth terms $\termthree$ and $\termfour$ involve mixed terms where one group (either $\bx$ or $\by$) contains only unobserved samples while the corresponding samples in the other group are all observed.

In order to bound $\mmdu{\bx}{\by}$, we propose to bound the four terms $\termone, \termtwo, \termthree$ and $\termfour$ separately. Given that $\termtwo$ includes the observed samples only, we focus on bounding $\termone, \termtwo$ and $\termfour$. For the first term $\termone$, where all samples are unobserved, it can be seen that $1 \ge \kernel{\x}{\y} > 0$ for any $\x,\y \in \mathbb{R}$. Hence, we have ${\mln{1}(\mln{1} - 1)}/{2} \ge \sum_{\iconstant = 1}^{\mln{1}} \sum_{\jconstant = \iconstant + 1}^{\mln{1}} \kernel{\xln{\iconstant}}{\xln{\jconstant}} > 0$, ${\mln{2}(\mln{2} - 1)}/{2} \ge \sum_{\iconstant = 1}^{\mln{2}} \sum_{\jconstant = \iconstant +1}^{\mln{2}} \kernel{\yln{\iconstant}}{\yln{\jconstant}} > 0$, and $\mln{1}\mln{2} \ge \sum_{\iconstant = 1}^{\mln{1}} \sum_{\jconstant = 1}^{\mln{2}} \kernel{\xln{\iconstant}}{\yln{\jconstant}} > 0$, which then according to the definition of $\termone$ in Lemma \ref{lemma:1}, follows
\begin{align}\label{eqn:bounding termone}
 \frac{\mln{1}(\mln{1} - 1)}{\nln{1}(\nln{1}-1)} + \frac{\mln{2}(\mln{2} - 1)}{\nln{2}(\nln{2}-1)} > \termone > - \frac{2}{\nln{1}\nln{2}} \mln{1}\mln{2}.
\end{align}

The primary challenge of bounding $\mmdu{\bx}{\by}$ lies in providing tight bounds for terms $\termthree$ and $\termfour$, where both observed and unobserved samples are presented. We focus on the term $\termthree$, with the understanding that a similar approach can be applied to the term $\termfour$. To start, we notice that the term $\termthree$ can be rewritten as
\begin{align} \label{rewrite:termthree:1} 
	\termthree & =  \sum_{\iconstant = 1}^{\mln{1}} \left(\constant{1} \sum_{\jconstant = \mln{1} + 1}^{\nln{1}} \kernel{\xln{\iconstant}}{\xln{\jconstant}} - 	\constant{3} \sum_{\jconstant = \mln{2} + 1}^{\nln{2}} \kernel{\xln{\iconstant}}{\yln{\jconstant}} \right),
\end{align}
where $\constant{1} = \frac{2}{\nln{1}(\nln{1}-1)}$ and $\constant{3} = \frac{2}{\nln{1}\nln{2}}$. Here, we introduce $\btln{1}(\z)$, defined as 
\begin{align*}
    \btln{1}(\z) = \constant{1} \sum_{\jconstant = \mln{1} + 1}^{\nln{1}} \kernel{\z}{\xln{\jconstant}} - 	\constant{3} \sum_{\jconstant = \mln{2} + 1}^{\nln{2}} \kernel{\z}{\yln{\jconstant}}.
\end{align*}  
Thus, $\termthree$ can be further expressed as:
\begin{align} \label{rewrite:termthree:2}
		\termthree = \sum_{\iconstant = 1}^{\mln{1}} \btln{1}(\xln{\iconstant}),
\end{align}
i.e. $\termthree$ is the summation of the function values $\btln{1}(\z)$ applied to all unobserved samples in $\bx$. This expression implies that in order to bound $\termthree$, it is sufficient to study bounds of function $\btln{1}(\z)$. We proceed by examining $\btln{1}(\z)$'s behavior, using the following lemma:

\subsection{Proof of Lemma \ref{lemma:2}} \label{Proof of Lemma 2}
\setcounter{lemma}{0}
\begin{lemma} \label{supp:lemma:2}
	Let $\xln{1}, \ldots, \xln{\ellsub{1}}$ and $\yln{1}, \ldots, \yln{\ellsub{2}}$ be univariate real values, that are observed. Suppose $\saln{1}, \ldots, \saln{\ellsub{1}}$, $\bbln{1}, \ldots, \bbln{\ellsub{2}}$ and $\beta$ are positive constants. Define
	\begin{align*}
		\bt(\z) = \sum_{\iconstant = 1}^{\ellsub{1}} \saln{\iconstant} \exp(-\beta|\xln{\iconstant} - \z|) - \sum_{\iconstant = 1}^{\ellsub{2}} \bbln{\iconstant} \exp(-\beta|\yln{\iconstant} - \z|)
	\end{align*}
	as a function of $\z \in \mathbb{R}$. Subsequently, for any given $\zln{0} \in \mathbb{R}$,
	\begin{align} \label{supp:lemma:2:eqn:1} 
		\bt(\zln{0}) \ge \min \{0, \bt(\xln{1}),\ldots,  \bt(\xln{\ellsub{1}}), \bt(\yln{1}),\ldots,  \bt(\yln{\ellsub{2}})\}.
	\end{align}
	On the other hand,
	\begin{align} \label{supp:lemma:2:eqn:2}
		\bt(\zln{0}) \le \max \{0, \bt(\xln{1}),\ldots,  \bt(\xln{\ellsub{1}}), \bt(\yln{1}),\ldots,  \bt(\yln{\ellsub{2}})\}.
	\end{align}
\end{lemma}

\begin{proof}
	We will first prove that inequality \eqref{supp:lemma:2:eqn:1} holds. Let $\zln{1} =  \min\{\xln{1}, \ldots, \xln{\ellsub{1}}, \yln{1}, \ldots, \yln{\ellsub{2}}\}$. When $\z < \zln{1} $, it follows
	\begin{align*}
		\bt(\z) &= \sum_{\iconstant = 1}^{\ellsub{1}} \saln{\iconstant} \exp(-\beta(\xln{\iconstant} - \z)) - \sum_{\iconstant = 1}^{\ellsub{2}} \bbln{\iconstant} \exp(-\beta(\yln{\iconstant} - \z))\\
		& =  \sum_{\iconstant = 1}^{\ellsub{1}} \saln{\iconstant} \exp(-\beta(\xln{\iconstant} - \zln{1} +\zln{1} -\z )) - \sum_{\iconstant = 1}^{\ellsub{2}} \bbln{\iconstant} \exp(-\beta(\yln{\iconstant} - \zln{1} +\zln{1} -\z))\\
		& = \sum_{\iconstant = 1}^{\ellsub{1}} \saln{\iconstant} \exp(-\beta(\xln{\iconstant} - \zln{1})) \exp(-\beta(\zln{1} -\z)) - \sum_{\iconstant = 1}^{\ellsub{2}} \bbln{\iconstant} \exp(-\beta(\yln{\iconstant} - \zln{1})) \exp(-\beta(\zln{1} -\z))\\
		& = \exp(-\beta(\zln{1} -\z)) \left\{\sum_{\iconstant = 1}^{\ellsub{1}} \saln{\iconstant} \exp(-\beta(\xln{\iconstant} - \zln{1})) - \sum_{\iconstant = 1}^{\ellsub{2}} \bbln{\iconstant} \exp(-\beta(\yln{\iconstant} - \zln{1})) \right\}\\
		& =  \exp(-\beta(\zln{1} -\z)) \bt(\zln{1}).
	\end{align*}
	Notice that, since $\beta > 0$ and $\z < \zln{1}$,
	\begin{align*}
		0 < \exp(-\beta(\zln{1} -\z)) < 1.
	\end{align*}
	Hence, if $\bt(\zln{1}) \ge 0$,
	\begin{align*}
		&0 \le \exp(-\beta(\zln{1} -\z)) \bt(\zln{1}) \le \bt(\zln{1});\\
		\Rightarrow & 0 \le \bt(\z)
	\end{align*}
	if $\bt(\zln{1}) < 0$,
	\begin{align*}
		&0 > \exp(-\beta(\zln{1} -\z)) \bt(\zln{1}) > \bt(\zln{1});\\
		\Rightarrow & \bt(\z) > \bt(\zln{1}).
	\end{align*}
	Thus, for any given $\zln{0} \in \mathbb{R}$, if $\zln{0} \in (-\infty, \zln{1})$, we have
	\begin{align*}
		\bt(\zln{0}) \ge \min\{0, \bt(\zln{1})\},
	\end{align*}
	which proves \eqref{supp:lemma:2:eqn:1} when $\zln{0} \in (-\infty, \zln{1})$.

	Similarly, let $\zln{2} =  \max\{\xln{1}, \ldots, \xln{\ellsub{1}}, \yln{1}, \ldots, \yln{\ellsub{2}}\}$. When $\z > \zln{2} $, it follows
	\begin{align*}
		\bt(\z) &= \sum_{\iconstant = 1}^{\ellsub{1}} \saln{\iconstant} \exp(-\beta(\z - \xln{\iconstant})) - \sum_{\iconstant = 1}^{\ellsub{2}} \bbln{\iconstant} \exp(-\beta(\z - \yln{\iconstant}))\\
		& =  \sum_{\iconstant = 1}^{\ellsub{1}} \saln{\iconstant} \exp(-\beta(\z - \zln{2} + \zln{2} - \xln{\iconstant})) - \sum_{\iconstant = 1}^{\ellsub{2}} \bbln{\iconstant} \exp(-\beta(\z - \zln{2} + \zln{2} - \yln{\iconstant}))\\
		& = \sum_{\iconstant = 1}^{\ellsub{1}} \saln{\iconstant} \exp(-\beta(\zln{2}-\xln{\iconstant} )) \exp(-\beta(\z -\zln{2})) - \sum_{\iconstant = 1}^{\ellsub{2}} \bbln{\iconstant} \exp(-\beta(\zln{2} -\yln{\iconstant} )) \exp(-\beta(\z -\zln{2}))\\
		& = \exp(-\beta(\z -\zln{2})) \left\{\sum_{\iconstant = 1}^{\ellsub{1}} \saln{\iconstant} \exp(-\beta(\zln{2} - \xln{\iconstant})) - \sum_{\iconstant = 1}^{\ellsub{2}} \bbln{\iconstant} \exp(-\beta(\zln{2} - \yln{\iconstant} )) \right\}\\
		& =  \exp(-\beta(\z -\zln{2})) \bt(\zln{2}).
	\end{align*}
	Notice that, since $\beta > 0$ and $\z > \zln{2}$,
	\begin{align*}
		0 < \exp(-\beta(\z -\zln{2})) < 1.
	\end{align*}
	Hence, if $\bt(\zln{2}) \ge 0$,
	\begin{align*}
		&0 \le \exp(-\beta(\z -\zln{2})) \bt(\zln{2}) \le \bt(\zln{2});\\
		\Rightarrow & 0 \le \bt(\z)
	\end{align*}
	if $\bt(\zln{1}) < 0$,
	\begin{align*}
		&0 > \exp(-\beta(\z -\zln{2})) \bt(\zln{2}) > \bt(\zln{2});\\
		\Rightarrow & \bt(\z) > \bt(\zln{2}).
	\end{align*}
	Thus, for any given $\zln{0} \in \mathbb{R}$, if $\zln{0} \in (\zln{2}, \infty)$, we have
	\begin{align*}
		\bt(\zln{0}) > \min\{0, \bt(\zln{2})\},
	\end{align*}
	which proves \eqref{supp:lemma:2:eqn:1} when $\zln{0} \in (\zln{2}, \infty)$.

	Suppose $\zln{1} < \zln{0} < \zln{2}$ and $\zln{0} \notin \{\xln{1}, \ldots, \xln{\ellsub{1}}, \yln{1}, \ldots, \yln{\ellsub{2}}\}$. Then, there must be at least one number in $\{\xln{1}, \ldots, \xln{\ellsub{1}}, \yln{1}, \ldots, \yln{\ellsub{2}}\}$ smaller than $\z$, and at least one number larger than $\z$. Subsequently, denote $\zln{3}$ as the maximum number in $\{\xln{1}, \ldots, \xln{\ellsub{1}}, \yln{1}, \ldots, \yln{\ellsub{2}}\}$ smaller than $\zln{0}$; denote $\zln{4}$ as the minimum number in $\{\xln{1}, \ldots, \xln{\ellsub{1}}, \yln{1}, \ldots, \yln{\ellsub{2}}\}$ larger than $\zln{0}$.
	
	Notice that
	\begin{align*}
		\bt(\z) & = \sum_{\iconstant = 1}^{\ellsub{1}} \saln{\iconstant} \exp(-\beta|\xln{\iconstant} - \z|) - \sum_{\iconstant = 1}^{\ellsub{2}} \bbln{\iconstant} \exp(-\beta|\yln{\iconstant} - \z|)\\
		& = \sum_{\iconstant = 1}^{\ellsub{1}} \saln{\iconstant} \indicator{\xln{\iconstant} \le \zln{3}} \exp(-\beta  |\xln{\iconstant} - \z| ) + \sum_{\iconstant = 1}^{\ellsub{1}} \saln{\iconstant} \indicator{\xln{\iconstant} > \zln{3}} \exp(-\beta  |\xln{\iconstant} - \z| )\\
		& - \sum_{\iconstant = 1}^{\ellsub{2}} \bbln{\iconstant} \indicator{\yln{\iconstant} \le \zln{3}} \exp(-\beta|\yln{\iconstant} - \z|) - \sum_{\iconstant = 1}^{\ellsub{2}} \bbln{\iconstant} \indicator{\yln{\iconstant} > \zln{3}} \exp(-\beta|\yln{\iconstant} - \z|).
	\end{align*}
	When $\zln{3} < \z < \zln{4}$, according to the definition of $\zln{3}$ and $\zln{4}$, for any $\xln{\iconstant}, \yln{\iconstant} \le \zln{3}$, it follows $\z > \xln{\iconstant}, \yln{\iconstant}$. On the other hand, for for any $\xln{\iconstant}, \yln{\iconstant} > \zln{3}$, it follows $\xln{\iconstant}, \yln{\iconstant} \ge \zln{4}$, which gives $\z < \xln{\iconstant}, \yln{\iconstant}$.
	
	Hence, when $\zln{3} < \z < \zln{4}$,
	\begin{align*}
		\bt(\z) & = \sum_{\iconstant = 1}^{\ellsub{1}} \saln{\iconstant} \indicator{\xln{\iconstant} \le \zln{3}} \exp(-\beta  (\z - \xln{\iconstant}) ) + \sum_{\iconstant = 1}^{\ellsub{1}} \saln{\iconstant} \indicator{\xln{\iconstant} > \zln{3}} \exp(-\beta  (\xln{\iconstant} - \z) )\\
		& - \sum_{\iconstant = 1}^{\ellsub{2}} \bbln{\iconstant} \indicator{\yln{\iconstant} \le \zln{3}} \exp(-\beta(\z - \yln{\iconstant}) ) - \sum_{\iconstant = 1}^{\ellsub{2}} \bbln{\iconstant} \indicator{\yln{\iconstant} > \zln{3}} \exp(-\beta(\yln{\iconstant} - \z)).\\
	\end{align*}
	Further,
	\begin{align*}
		\bt(\z) & = \sum_{\iconstant = 1}^{\ellsub{1}} \saln{\iconstant} \indicator{\xln{\iconstant} \le \zln{3}} \exp(-\beta  (\z - \zln{3} +\zln{3} - \xln{\iconstant}) ) + \sum_{\iconstant = 1}^{\ellsub{1}} \saln{\iconstant} \indicator{\xln{\iconstant} > \zln{3}} \exp(-\beta  (\xln{\iconstant} -\zln{4} + \zln{4} - \z) )\\
		& - \sum_{\iconstant = 1}^{\ellsub{2}} \bbln{\iconstant} \indicator{\yln{\iconstant} \le \zln{3}} \exp(-\beta(\z - \zln{3} + \zln{3} - \yln{\iconstant}) ) - \sum_{\iconstant = 1}^{\ellsub{2}} \bbln{\iconstant} \indicator{\yln{\iconstant} > \zln{3}} \exp(-\beta(\yln{\iconstant} -\zln{4} + \zln{4} - \z)).\\
		&= \exp(-\beta(\z - \zln{3}))  \sum_{\iconstant = 1}^{\ellsub{1}}\saln{\iconstant} \indicator{\xln{\iconstant} \le \zln{3}} \exp(-\beta  (\zln{3} - \xln{\iconstant}) ) \\
		& + \exp(-\beta(\zln{4} - \z))  \sum_{\iconstant = 1}^{\ellsub{1}} \saln{\iconstant} \indicator{\xln{\iconstant} > \zln{3}} \exp(-\beta  (\xln{\iconstant} -\zln{4}) )\\
		& - \exp(-\beta(\z - \zln{3})) \sum_{\iconstant = 1}^{\ellsub{2}} \bbln{\iconstant} \indicator{\yln{\iconstant} \le \zln{3}} \exp(-\beta(\zln{3} -\yln{\iconstant}) )\\
		& - \exp(-\beta(\zln{4} - \z)) \sum_{\iconstant = 1}^{\ellsub{2}} \bbln{\iconstant} \indicator{\yln{\iconstant} > \zln{3}} \exp(-\beta(\yln{\iconstant} -\zln{4} )).\\
		&= \exp(-\beta(\z - \zln{3})) \left\{ \sum_{\iconstant = 1}^{\ellsub{1}}\saln{\iconstant} \indicator{\xln{\iconstant} \le \zln{3}} \exp(-\beta  (\zln{3} - \xln{\iconstant}) ) - \sum_{\iconstant = 1}^{\ellsub{2}} \bbln{\iconstant} \indicator{\yln{\iconstant} \le \zln{3}} \exp(-\beta(\zln{3} -\yln{\iconstant}) ) \right\}\\
		& + \exp(-\beta(\zln{4} - \z)) \left\{ \sum_{\iconstant = 1}^{\ellsub{1}} \saln{\iconstant} \indicator{\xln{\iconstant} > \zln{3}} \exp(-\beta  (\xln{\iconstant} -\zln{4}) ) - \sum_{\iconstant = 1}^{\ellsub{2}} \bbln{\iconstant} \indicator{\yln{\iconstant} > \zln{3}} \exp(-\beta(\yln{\iconstant} -\zln{4} )) \right\}.
	\end{align*}
	For notation ease, let us denote
	\begin{align*}
		&A := \left\{ \sum_{\iconstant = 1}^{\ellsub{1}}\saln{\iconstant} \indicator{\xln{\iconstant} \le \zln{3}} \exp(-\beta  (\zln{3} - \xln{\iconstant}) ) - \sum_{\iconstant = 1}^{\ellsub{2}} \bbln{\iconstant} \indicator{\yln{\iconstant} \le \zln{3}} \exp(-\beta(\zln{3} -\yln{\iconstant}) ) \right\},\\
		&B := \left\{ \sum_{\iconstant = 1}^{\ellsub{1}} \saln{\iconstant} \indicator{\xln{\iconstant} > \zln{3}} \exp(-\beta  (\xln{\iconstant} -\zln{4}) ) - \sum_{\iconstant = 1}^{\ellsub{2}} \bbln{\iconstant} \indicator{\yln{\iconstant} > \zln{3}} \exp(-\beta(\yln{\iconstant}-\zln{4}) ) \right\}.
	\end{align*}
	Thus,
	\begin{align*}
		\bt(\z) = A\exp(-\beta(\z - \zln{3})) + B\exp(-\beta(\zln{4} - \z)) .
	\end{align*}
	Notice that if $B = 0$, for any $\zln{3} < \z < \zln{4}$, $\bt(\z) = A\exp(-\beta(\z - \zln{3}))$ is a monotonic increasing function in $(\zln{3},\zln{4})$ when $A < 0$, or a monotonic decreasing function in $(\zln{3},\zln{4})$ when $A > 0$, or a constant function in $(\zln{3},\zln{4})$ when $A = 0$. Thus, we always have $\min\{\bt(\zln{3}), \bt(\zln{4}) \} \le \bt(\zln{0})$ for $\zln{0} \in (\zln{3}, \zln{4})$, which proves inequality \eqref{supp:lemma:2:eqn:1} directly. 
	
	If, however, $B \neq 0$, let $\bt(\z) = 0$, we have
	\begin{align*}
		&A\exp(-\beta(\z - \zln{3})) = -B\exp(-\beta(\zln{4} - \z))\\
		\Rightarrow &-\frac{A}{B} = \frac{\exp(-\beta(\zln{4} - \z))}{\exp(-\beta(\z - \zln{3}))}.
	\end{align*}
	When $-\frac{A}{B} \le 0$, since $\frac{\exp(-\beta(\zln{4} - \z))}{\exp(-\beta(\z - \zln{3}))} > 0$, then $\bt(\z)$ cannot take 0 in $(\zln{3}, \zln{4})$. If however, $-\frac{A}{B} > 0$, we further have
	\begin{align*}
		&\log \left(-\frac{A}{B}\right) =  -\beta(\zln{4} - \z) - \{-\beta(\z - \zln{3})\}\\
		\Rightarrow&\log \left(-\frac{A}{B}\right) =  -\beta\zln{4} + \beta\z + \beta \z - \beta \zln{3}\\
		\Rightarrow&\z = \frac{1}{\beta}\log \left(-\frac{A}{B}\right) + \zln{4} + \zln{3}.
	\end{align*}
	That is, $\bt(\z) = 0$ when $\z = \frac{1}{\beta}\log \left(-\frac{A}{B}\right) + \zln{4} + \zln{3} \in (\zln{3}, \zln{4})$. Overall, we have shown $\bt(\z)$ takes 0 at most once in $(\zln{3}, \zln{4})$ when $B \neq 0$. This result will be used subsequently for proving our final conclusion.

	Now, taking derivative of $\bt(\z)$ of $\z$,
	\begin{align*}
		\frac{\partial \bt(\z)}{\partial \z} &= -\beta A\exp(-\beta(\z - \zln{3})) + \beta B \exp(-\beta(\zln{4} - \z)).
	\end{align*}
	Further, taking derivative of $\frac{\partial \bt(\z)}{\partial \z}$ of $\z$,
	\begin{align*}
		\frac{\partial^2 \bt(\z)}{\partial^2 \z} &= \beta^2 A\exp(-\beta(\z - \zln{3})) + \beta^2 B \exp(-\beta(\zln{4} - \z))\\
		& = \beta^2 \bt(\z).
	\end{align*}
	
	We are going to prove $\bt(\zln{0}) \ge \min\{0, \bt{(\zln{3})}, \bt{(\zln{4})}\}$ when $B \neq 0$ using contradiction. Let us assume that
	\begin{align} \label{supp:lemma:2:eqn:3}
		\bt(\zln{0}) < \min\{0, \bt{(\zln{3})}, \bt{(\zln{4})}\}.
	\end{align}
	Then, since $\beta >0$, it must have
	\begin{align*}
		\left.\frac{\partial^2 \bt(\z)}{\partial^2 \z}\right\vert_{\zln{0}} = \beta^2 \bt(\zln{0}) < 0.
	\end{align*}
	For $\left.\frac{\partial \bt(\z)}{\partial \z}\right\vert_{\zln{0}}$, it is either $(i): \left.\frac{\partial \bt(\z)}{\partial \z}\right\vert_{\zln{0}} < 0$ or $(ii): \left.\frac{\partial \bt(\z)}{\partial \z}\right\vert_{\zln{0}} \ge 0$. We are going to discuss the two cases separately.
	
	Suppose $(i): \left.\frac{\partial \bt(\z)}{\partial \z}\right\vert_{\zln{0}} < 0$. If for any $\z \in [\zln{0}, \zln{4})$,  $\frac{\partial^2 \bt(\z)}{\partial^2 \z} < 0$. Then, $\frac{\partial \bt(\z)}{\partial \z}$ is a monotonic decreasing function of $\z$ in $[\zln{0}, \zln{4})$. Hence, for any $\z \in [\zln{0}, \zln{4})$,
	\begin{align*}
		\frac{\partial \bt(\z)}{\partial \z} < \left.\frac{\partial \bt(\z)}{\partial \z}\right\vert_{\zln{0}} < 0
	\end{align*}
	Thus, $\bt(\z)$ is a monotonic decreasing function in $\z \in [\zln{0}, \zln{4})$, giving us
	\begin{align*}
		\bt(\zln{0}) > \bt(\zln{4}), 
	\end{align*}
	which is contradicted to the assumption \eqref{supp:lemma:2:eqn:3}.
	
	If, however, for any $\z \in [\zln{0}, \zln{4})$, $\frac{\partial^2 \bt(\z)}{\partial^2 \z} < 0$ does not holds, recall that we have shown $\bt(\z)$ take 0 at most once in $(\zln{3}, \zln{4})$ when $B \neq 0$. This then implies there is only one point 
	$\st \in [\zln{0}, \zln{4})$ such that $\left.\frac{\partial^2 \bt(\z)}{\partial^2 \z}\right\vert_{\st} = 0$. Since $\frac{\partial^2 \bt(\z)}{\partial^2 \z}$ is a continuous function and $\left.\frac{\partial^2 \bt(\z)}{\partial^2 \z}\right\vert_{\zln{0}} < 0$, we have
	\begin{align*}
		\frac{\partial^2 \bt(\z)}{\partial^2 \z} < 0,~\z \in [\zln{0},\st).
	\end{align*}
	Thus, $\frac{\partial \bt(\z)}{\partial \z}$ is a monotonic decreasing function in $\z \in [\zln{0}, \st)$, giving us
	\begin{align}
		\frac{\partial \bt(\z)}{\partial \z} < \left.\frac{\partial \bt(\z)}{\partial \z}\right\vert_{\zln{0}} < 0,~\z \in [\zln{0},\st).  \label{supp:lemma:2:eqn:4}
	\end{align}
	However, notice that we also have
	\begin{align*}
		&\left.\frac{\partial^2 \bt(\z)}{\partial^2 \z}\right\vert_{\st} = 0 > \left.\frac{\partial^2 \bt(\z)}{\partial^2 \z}\right\vert_{\zln{0}}, \nonumber\\
		\Rightarrow & \bt(\st) = 0 > \bt(\zln{0}) \nonumber\\
		\Rightarrow & \exists \st' \in (\zln{0}, \st)~\text{such that}  \left.\frac{\partial \bt(\z)}{\partial \z}\right\vert_{\st'} = \frac{\bt(\st) - \bt(\zln{0})}{\st - \zln{0}} > 0,  ~~~(\text{mean value theorem}) 
	\end{align*}
	which contradicts \eqref{supp:lemma:2:eqn:4}.
	
	Suppose $(ii): \left.\frac{\partial \bt(\z)}{\partial \z}\right\vert_{\zln{0}} \ge 0$. The prove method is similar as to when $(i)$ holds.
	
	If for any $\z \in (\zln{3}, \zln{0}]$,  $\frac{\partial^2 \bt(\z)}{\partial^2 \z} < 0$. Then, $\frac{\partial \bt(\z)}{\partial \z}$ is a monotonic decreasing function of $\z$ in $(\zln{3}, \zln{0}]$. Hence, for any $\z \in (\zln{3}, \zln{0}]$, 
	\begin{align*}
		\frac{\partial \bt(\z)}{\partial \z} > \left.\frac{\partial \bt(\z)}{\partial \z}\right\vert_{\zln{0}} \ge 0.
	\end{align*}
	Thus, $\bt(\z)$ is a monotonic increasing function in $\z \in (\zln{3}, \zln{0}]$, giving
	\begin{align*}
		\bt(\zln{3}) < \bt(\zln{0}),
	\end{align*}
	which contradicts to the assumption \eqref{supp:lemma:2:eqn:3}.
	
	If, however, for any $\z \in (\zln{3}, \zln{0}]$,  $\frac{\partial^2 \bt(\z)}{\partial^2 \z} < 0$ does not holds, recall that we have shown $\bt(\z)$ taken 0 at most once in $(\zln{3}, \zln{4})$. This then implies there is only one point 
	$\st \in (\zln{3}, \zln{0}]$ such that $\left.\frac{\partial^2 \bt(\z)}{\partial^2 \z}\right\vert_{\st} = 0$. Since $\frac{\partial^2 \bt(\z)}{\partial^2 \z}$ is a continuous function and $\left.\frac{\partial^2 \bt(\z)}{\partial^2 \z}\right\vert_{\zln{0}} < 0$, we have
	\begin{align*}
		\frac{\partial^2 \bt(\z)}{\partial^2 \z} < 0,~\z \in (\st,\zln{0}].
	\end{align*}
	Thus, $\frac{\partial \bt(\z)}{\partial \z}$ is a monotonic decreasing function in $\z \in (\st,\zln{0}]$, giving us
	\begin{align}
		\frac{\partial \bt(\z)}{\partial \z} > \left.\frac{\partial \bt(\z)}{\partial \z}\right\vert_{\zln{0}} \ge 0,~\z \in (\st,\zln{0}].  \label{supp:lemma:2:eqn:5}
	\end{align}
	However, notice that we also have
	\begin{align*}
		&\left.\frac{\partial^2 \bt(\z)}{\partial^2 \z}\right\vert_{\st} = 0 > \left.\frac{\partial^2 \bt(\z)}{\partial^2 \z}\right\vert_{\zln{0}}, \nonumber\\
		\Rightarrow & \bt(\st) = 0 > \bt(\zln{0}) \nonumber\\
		\Rightarrow & \exists \st' \in (\st, \zln{0})~\text{such that}  \left.\frac{\partial \bt(\z)}{\partial \z}\right\vert_{\st'} = \frac{\bt(\st) - \bt(\zln{0})}{\st - \zln{0}} < 0.  ~~~(\text{mean value theorem})
	\end{align*}
	which contradicts \eqref{supp:lemma:2:eqn:5}. Hence, we finish our prove for inequality \eqref{supp:lemma:2:eqn:1}.
	
	We now prove inequality \eqref{supp:lemma:2:eqn:2}. Notice that 
	\begin{align*}
		- \bt(\z) = \sum_{\iconstant = 1}^{\ellsub{2}} \bbln{\iconstant} \exp(-\beta|\yln{\iconstant} - \z|) - \sum_{\iconstant = 1}^{\ellsub{1}} \saln{\iconstant} \exp(-\beta|\xln{\iconstant} - \z|).
	\end{align*}
	Subsequently, using the result of inequality \eqref{supp:lemma:2:eqn:1}, we have
	\begin{align*}
		&- \bt(\zln{0}) \ge \min \{0, -\bt(\xln{1}),\ldots,  -\bt(\xln{\ellsub{1}}), -\bt(\yln{1}),\ldots,  - \bt(\yln{\ellsub{2}})\}\\
		\Rightarrow &\bt(\zln{0})  \le - \min \{0, -\bt(\xln{1}),\ldots,  -\bt(\xln{\ellsub{1}}), -\bt(\yln{1}),\ldots,  - \bt(\yln{\ellsub{2}})\} \\
		\Rightarrow &\bt(\zln{0})  \le \max \{0, \bt(\xln{1}),\ldots,  \bt(\xln{\ellsub{1}}), \bt(\yln{1}),\ldots,   \bt(\yln{\ellsub{2}})\}, 
	\end{align*}
	which proves inequality \eqref{supp:lemma:2:eqn:2}.
\end{proof}
This lemma provides a linear time algorithm for computing the bounds of $\bt(\z)$, be checking all the function values of $\bt(\z)$ to the set $\{\xln{1}, \ldots, \xln{\ellsub{1}}, \yln{1}, \ldots, \yln{\ellsub{2}}\}$ and comparing all these function values with 0.
Notice that $\btln{1}(\z)$ is a special form of $\bt(\z)$ by letting all $\saln{1}, \ldots, \saln{\ellsub{1}}$ in Lemma \ref{lemma:2} equal to $\constant{1}$ and all  $\bbln{1}, \ldots, \bbln{\ellsub{2}}$  equal to $\constant{2}$, where $\constant{1} = \frac{2}{\nln{1}(\nln{1}-1)}$, 
$\constant{2} = \frac{2}{\nln{2}(\nln{2}-1)}$. Hence, Lemma \ref{lemma:2} allows us to compute the bounds of $\btln{1}(\z)$ efficiently, which then gives the bounds of $\termthree$ according to Equation \eqref{rewrite:termthree:2}. This result, combining with Equation \eqref{eqn:bounding termone}, allows us to conclude the bounds of MMD for univariate samples, which is proved in Theorem \ref{supp:theorem:1}.

\subsection{Proof of Lemma \ref{lemma:4}} \label{Proof of Lemma 4}

This section proves Lemma \ref{lemma:4}, which extends Lemma \ref{lemma:2} into higher dimensions. 

\begin{lemma} \label{supp:lemma:4}
	Let $\xln{1}, \ldots, \xln{\ellsub{1}}, \yln{1}, \ldots, \yln{\ellsub{2}} \in \mathbb{R}^d$
	be values that are fully observed.
	Suppose $\saln{1}, \ldots, \saln{\ellsub{1}}$, $\bbln{1}, \ldots, \bbln{\ellsub{2}}, \beta$ 
	are positive constants. For $\z = (\z(1), \ldots, \z(\dd) ) \in \mathbb{R}^d$  
	with missing components, define
	\begin{align*}
		\bt( \{ \z (j) \,:\, j \in \buln{z} \}) 
		= \sum_{\iconstant = 1}^{\ellsub{1}} \saln{\iconstant}  
		\exp\left(-\beta \sum_{\jconstant \in \buln{\z}} | \xln{\iconstant}(\jconstant) 
		- \z(\jconstant)|\right) 
		- \sum_{\iconstant = 1}^{\ellsub{2}} \bbln{\iconstant} 
		\exp\left(-\beta  \sum_{\jconstant \in \buln{\z}} |\yln{\iconstant}(\jconstant) - \z(\jconstant)|\right)
	\end{align*}
	as a function of the unobserved components of $\z$ and let 
	\begin{align*}
		\setchi = \{ \bt( \{ \z(j) \,:\, j \in \buln{z} \}) \ :\, 
		\z(\iconstant)\in \{\xln{1}(\iconstant), \ldots, \xln{\ellsub{1}}(\iconstant), 
		\yln{1}(\iconstant), \ldots, \yln{\ellsub{2}}(\iconstant) \}, \iconstant \in \buln{z} \}.
	\end{align*}
	Then, for any possible imputation $z^{\ast}$ of $\z$, 
	\begin{align} \label{supp:lemma:4:eqn:1}
		\bt( \{ z^{\ast}(j) \,:\, j \in \buln{z} \}) \ge \min \{0, \min\setchi\},
	\end{align}
	and
	\begin{align} \label{supp:lemma:4:eqn:2}
		\bt( \{ z^{\ast}(j) \,:\, j \in \buln{z} \})  \le \max \{0, \max\setchi\}.
	\end{align}
\end{lemma}

\begin{proof}
	We will only prove inequality \eqref{supp:lemma:4:eqn:1} and inequality \eqref{supp:lemma:4:eqn:2} can be proved following the same method.
	
	Let us prove inequality \eqref{supp:lemma:4:eqn:1} using mathematical induction. When $|\buln{z}| = 1$, this lemma is exactly Lemma \ref{supp:lemma:2}, which has already been shown. Suppose when $|\buln{z}| = \lconstant - 1$, where $\lconstant$ is any integer such that $1 \le \lconstant$,  inequality \eqref{supp:lemma:4:eqn:1} holds. We are going to show when $|\buln{z}| = \lconstant$, inequality \eqref{supp:lemma:4:eqn:1} still holds.
	
	Without loss of generality, let us assume (after relabeling) that $\buln{\z} = 1, \ldots, \lconstant$. Subsequently, for any given $(\zsln{1}, \ldots, \zsln{\lconstant}) \in \mathbb{R}^{\lconstant}$, it follows
	\begin{align*}
		\bt(\zsln{1}, \ldots, \zsln{\lconstant}) &= \sum_{\iconstant = 1}^{\ellsub{1}} \saln{\iconstant}  \exp\left(-\beta \sum_{\jconstant = 1}^{\lconstant} |\xln{\iconstant}(\jconstant) - \zsln{\jconstant}|\right) - \sum_{\iconstant = 1}^{\ellsub{2}} \bbln{\iconstant} \exp\left(-\beta  \sum_{\jconstant = 1}^{\lconstant}|\yln{\iconstant}(\jconstant) - \zsln{\jconstant}|\right)\\
		& = \sum_{\iconstant = 1}^{\ellsub{1}} \saln{\iconstant}  \exp\left(-\beta |\xln{\iconstant}(1) - \zsln{1}|\right) \exp\left(-\beta \sum_{\jconstant = 2}^{\lconstant} |\xln{\iconstant}(\jconstant) - \zsln{\jconstant}|\right)\\
		& - \sum_{\iconstant = 1}^{\ellsub{2}} \bbln{\iconstant} \exp\left(-\beta  |\yln{\iconstant}(1) - \zsln{1}|\right) \exp\left(-\beta  \sum_{\jconstant = 2}^{\lconstant}|\yln{\iconstant}(\jconstant) - \zsln{\jconstant}|\right).
	\end{align*}
	Denote
	\begin{align*}
		\saln{\iconstant}' &= \saln{\iconstant}  \exp\left(-\beta |\xln{\iconstant}(1) - \zsln{1}|\right),~\iconstant = 1,\ldots, \ellsub{1},\\
		\bbln{\iconstant}' &= \bbln{\iconstant}  \exp\left(-\beta  |\yln{\iconstant}(1) - \zsln{1}|\right),~\iconstant = 1,\ldots, \ellsub{2}.
	\end{align*}
	Consider a new function 
	\begin{align*}
		\bt'(\zln{2}, \ldots, \zln{\lconstant}) = \sum_{\iconstant = 1}^{\ellsub{1}} \saln{\iconstant}'  \exp\left(-\beta \sum_{\jconstant = 2}^{\lconstant} |\xln{\iconstant}(\jconstant) - \zln{\jconstant}|\right) - \sum_{\iconstant = 1}^{\ellsub{2}} \bbln{\iconstant}' \exp\left(-\beta  \sum_{\jconstant = 2}^{\lconstant}|\yln{\iconstant}(\jconstant) - \zln{\jconstant}|\right).
	\end{align*}
	Let
	\begin{align*}
		\setchi' = \{\bt'(\zln{2}, \ldots, \zln{\lconstant})| \zln{\jconstant} \in \{\xln{1}(\jconstant), \ldots, \xln{\ellsub{1}}(\jconstant), \yln{1}(\jconstant), \ldots, \yln{\ellsub{2}}(\jconstant) \}, \jconstant = 2, \ldots, \lconstant \}.
	\end{align*}
	Then, using the assumption that when $|\buln{z}| = \lconstant - 1$, inequality \eqref{supp:lemma:4:eqn:1} holds, it follows that
	\begin{align*}
		\bt'(\zsln{2}, \ldots, \zsln{\lconstant}) &\ge \min \{0, \min \setchi'\},
	\end{align*}
	Notice that 
	\begin{align}  
		\begin{split} \label{supp:lemma:4:eqn:3}
			\bt'(\zln{2}, \ldots, \zln{\lconstant}) &= \sum_{\iconstant = 1}^{\ellsub{1}} \saln{\iconstant}'  \exp\left(-\beta \sum_{\jconstant = 2}^{\lconstant} |\xln{\iconstant}(\jconstant) - \zln{\jconstant}|\right) - \sum_{\iconstant = 1}^{\ellsub{2}} \bbln{\iconstant}' \exp\left(-\beta  \sum_{\jconstant = 2}^{\lconstant}|\yln{\iconstant}(\jconstant) - \zln{\jconstant}|\right)\\
			& =  \sum_{\iconstant = 1}^{\ellsub{1}} \saln{\iconstant}  \exp\left(-\beta |\xln{\iconstant}(1) - \zsln{1}|\right)  \exp\left(-\beta \sum_{\jconstant = 2}^{\lconstant} |\xln{\iconstant}(\jconstant) - \zln{\jconstant}|\right) \\
			& - \sum_{\iconstant = 1}^{\ellsub{2}} \bbln{\iconstant}  \exp\left(-\beta  |\yln{\iconstant}(1) - \zsln{1}|\right) \exp\left(-\beta  \sum_{\jconstant = 2}^{\lconstant}|\yln{\iconstant}(\jconstant) - \zln{\jconstant}|\right)\\
			& =\bt(\zsln{1}, \zln{2}, \ldots, \zln{\lconstant}).
		\end{split}
	\end{align}
	Hence,
	\begin{align*}
		\bt'(\zsln{2}, \ldots, \zsln{\lconstant}) = \bt(\zsln{1}, \ldots, \zsln{\lconstant}).
	\end{align*}
	We therefore have
	\begin{align} \label{supp:lemma:4:eqn:4}
		\bt(\zsln{1}, \ldots, \zsln{\lconstant}) &\ge \min \{0, \min \chi'\}.
	\end{align}

	Further, let us denote 
	\begin{align} \label{supp:lemma:4:eqn:5}
		(\zln{2}', \ldots, \zln{\lconstant}') = \argmin_{\zln{\iconstant} \in \{\xln{1}(\iconstant), \ldots, \xln{\ellsub{1}}(\iconstant), \yln{1}(\iconstant), \ldots, \yln{\ellsub{2}}(\iconstant) \}, \iconstant = 2, \ldots, \lconstant} \bt'(\zln{2}, \ldots, \zln{\lconstant}).
	\end{align}
	That is,
	\begin{align*}
		\bt'(\zln{2}', \ldots, \zln{\lconstant}') = \min \setchi'.
	\end{align*}
	Applying \eqref{supp:lemma:4:eqn:3}, it follows
	\begin{align} \label{supp:lemma:4:eqn:6}
		\bt(\zsln{1}, \zln{2}', \ldots, \zln{\lconstant}') = \bt'(\zln{2}', \ldots, \zln{\lconstant}') = \min \setchi'.
	\end{align}
	
	Denote
	\begin{align*}
		\saln{\iconstant}'' &= \saln{\iconstant}  \exp\left(-\beta \sum_{\jconstant = 2}^{\lconstant} |\xln{\iconstant}(\jconstant) - \zln{\jconstant}'|\right),~\iconstant = 1,\ldots, \ellsub{1},\\
		\bbln{\iconstant}'' &= \bbln{\iconstant}  \exp\left(-\beta  \sum_{\jconstant = 2}^{\lconstant}|\yln{\iconstant}(\jconstant) - \zln{\jconstant}'|\right)~\iconstant = 1,\ldots, \ellsub{2}.
	\end{align*}
	Consider a new function
	\begin{align*}
		\bt''(\zln{1}) = \sum_{\iconstant=1}^{\ellsub{1}} \saln{\iconstant}'' \exp\left(-\beta|\xln{\iconstant}(1) - \zln{1}|\right) - \sum_{\iconstant=1}^{\ellsub{1}} \bbln{\iconstant}'' \exp\left(-\beta|\yln{\iconstant}(1) - \zln{1}|\right).
	\end{align*}
	Let
	\begin{align*}
		\setchi'' = \{\bt''(\zln{1})| \zln{1} \in \{\xln{1}(1), \ldots, \xln{\ellsub{1}}(1), \yln{1}(1), \ldots, \yln{\ellsub{2}}(1) \} \}.
	\end{align*}
	Then, using the result that when $|\buln{z}| = 1$, \eqref{supp:lemma:4:eqn:1}  holds, it follows
	\begin{align}  \label{supp:lemma:4:eqn:7}
		\bt''(\zsln{1}) &\ge \min\{0, \min \setchi''\}.
	\end{align}
	Notice that 
	\begin{align*}
		\bt''(\zln{1}) &= \sum_{\iconstant=1}^{\ellsub{1}} \saln{\iconstant}'' \exp\left(-\beta|\xln{\iconstant}(1) - \zln{1}|\right) - \sum_{\iconstant=1}^{\ellsub{1}} \bbln{\iconstant}'' \exp\left(-\beta|\yln{\iconstant}(1) - \zln{1}|\right) \\
		& = \sum_{\iconstant=1}^{\ellsub{1}} \saln{\iconstant}  \exp\left(-\beta \sum_{\jconstant = 2}^{\lconstant} |\xln{\iconstant}(\jconstant) - \zln{\jconstant}'|\right) \exp\left(-\beta|\xln{\iconstant}(1) - \zln{1}|\right)\\
		& - \sum_{\iconstant=1}^{\ellsub{1}} \bbln{\iconstant}  \exp\left(-\beta  \sum_{\jconstant = 2}^{\lconstant}|\yln{\iconstant}(\jconstant) - \zln{\jconstant}'|\right) \exp\left(-\beta|\yln{\iconstant}(1) - \zln{1}|\right) \\
		& = \bt(\zln{1}, \zln{2}', \ldots, \zln{\lconstant}').
	\end{align*}
	According to \eqref{supp:lemma:4:eqn:6}, 
	\begin{align*}
		\bt''(\zsln{1}) =  \min \setchi'.
	\end{align*}
	Notice that since \eqref{supp:lemma:4:eqn:7} holds, we have
	\begin{align*}
		\min \setchi' \ge \min\{0, \min \setchi''\},
	\end{align*}
	put which back into \eqref{supp:lemma:4:eqn:4}, we further have
	\begin{align*}
		\bt(\zsln{1}, \ldots, \zsln{\lconstant}) &\ge \min \{0, \min \setchi'\} \ge \min \{0, \min \setchi''\}.
	\end{align*}
	Notice that
	\begin{align*}
		\min \setchi'' &= \min \{\bt''(\zln{1})| \zln{1} \in \{\xln{1}(1), \ldots, \xln{\ellsub{1}}(1), \yln{1}(1), \ldots, \yln{\ellsub{2}}(1) \} \} \\
		& = \min \{\bt(\zln{1}, \zln{2}', \ldots, \zln{\lconstant}')| \zln{1} \in \{\xln{1}(1), \ldots, \xln{\ellsub{1}}(1), \yln{1}(1), \ldots, \yln{\ellsub{2}}(1) \} \},
	\end{align*}
	where 
	\begin{align*}
		\zln{\iconstant}' \in \{\xln{1}(\iconstant), \ldots, \xln{\ellsub{1}}(\iconstant), \yln{1}(\iconstant), \ldots, \yln{\ellsub{2}}(\iconstant) \}, \iconstant = 2, \ldots, \lconstant
	\end{align*}
	according to its definition in \eqref{supp:lemma:4:eqn:5}.
	Hence,
	\begin{align*}
		\min \setchi'' \ge \min \chi \Rightarrow \bt(\zsln{1}, \ldots, \zsln{\lconstant}) \ge \min \{0, \min \setchi\},
	\end{align*}
	which proves inequality \eqref{supp:lemma:4:eqn:1}.
\end{proof}

\subsection{Proof of Lemma \ref{lemma:5}} \label{Proof of Lemma 5}

In Lemma \ref{lemma:2}, it is proved that in order to compute the bounds of 
$\bt( \{ \z (j) \,:\, j \in \buln{z} \})$ for $z \in \mathbb{R}^d$,
we only need to check the imputations of $\z$ where its missing components
are imputed using the components of 
$\xln{1}, \ldots, \xln{\ellsub{1}}, \yln{1}, \ldots, \yln{\ellsub{2}}$.
However, computing $\bt( \{ \z (j) \,:\, j \in \buln{z} \})$ for all possible 
imputations using Lemma \ref{lemma:4} is 
$(\ellsub{1} + \ellsub{2})^{|\buln{\xln{\iconstant}}|}$, which 
is exponential in the number of unobserved components of $\z$ and impractical
to compute. To address this practical computational challenge, we introduce the following lemma:

\begin{lemma} \label{supp:lemma:5}
	Following the notation and definitions in Lemma \ref{lemma:4}, denote 
	\begin{align*}
		\xhln{\iconstant}(\jconstant) &= \max \{ |\xln{\iconstant}(\jconstant) - \xln{1}(\jconstant)|, \ldots, |\xln{\iconstant}(\jconstant) - \xln{\ellsub{1}}(\jconstant)|, |\xln{\iconstant}(\jconstant) - \yln{1}(\jconstant)|, \ldots, |\xln{\iconstant}(\jconstant) - \yln{\ellsub{2}}(\jconstant)|\}
	\end{align*}
	for any $\iconstant \in \{1,\ldots,\ellsub{1}\}, \jconstant \in \buln{z}$; denote
	\begin{align*}
		\yhln{\iconstant}(\jconstant) &= \max \{ |\yln{\iconstant}(\jconstant) - \ellsub{1}(\jconstant)|, \ldots, |\yln{\iconstant}(\jconstant) - \xln{\nln{1}}(\jconstant)|, |\yln{\iconstant}(\jconstant) - \yln{1}(\jconstant)|, \ldots, |\yln{\iconstant}(\jconstant) - \yln{\ellsub{2}}(\jconstant)|\}
	\end{align*}
	for any $\iconstant \in \{1,\ldots,\ellsub{2}\}, \jconstant \in \buln{z}$. Subsequently,
	\begin{align}
		&\min \setchi \ge  \sum_{\iconstant = 1}^{\ellsub{1}} \saln{\iconstant} \exp\left(-\beta \sum_{ \jconstant \in \buln{z}} \xhln{\iconstant}(\jconstant) \right) - \sum_{\iconstant = 1}^{\ellsub{2}} \bbln{\iconstant}, \label{supp:lemma:5:eqn:1}\\
		&\max \setchi \le  \sum_{\iconstant = 1}^{\ellsub{1}} \saln{\iconstant} - \sum_{\iconstant = 1}^{\ellsub{2}} \bbln{\iconstant} \exp\left(-\beta \sum_{ \jconstant \in \buln{z}} \yhln{\iconstant}(\jconstant) \right). \label{supp:lemma:5:eqn:2}
	\end{align}
\end{lemma}

\begin{proof}
	We will only prove inequality \eqref{supp:lemma:5:eqn:1} and inequality \eqref{supp:lemma:5:eqn:2} can be proved similarly. 
	
	Without loss of generality, let us assume (after relabeling) $\buln{z} = 1, \ldots, \lconstant$, where $\lconstant = |\buln{z}|$. Subsequently, let 
	\begin{align*}
		(\zsln{1}, \ldots, \zsln{\lconstant}) = \argmin_{ \{\zln{\iconstant} \in \{\xln{1}(\iconstant), \ldots, \xln{\ellsub{1}}(\iconstant), \yln{1}(\iconstant), \ldots, \yln{\ellsub{2}}(\iconstant) \}, \iconstant = 1, \ldots, \lconstant \}} \bt(\zln{1}, \ldots,\zln{\lconstant}).
	\end{align*}
	That is,
	\begin{align*}
		\bt(\zsln{1}, \ldots,\zsln{\lconstant}) = \min \setchi.
	\end{align*}
	Notice that for any $\jconstant \in \{1,\ldots, \lconstant\}$, 
	\begin{align*}
		\zsln{\jconstant} \in \{\xln{1}(\jconstant), \ldots, \xln{\ellsub{1}}(\jconstant), \yln{1}(\jconstant), \ldots, \yln{\ellsub{2}}(\jconstant) \}
	\end{align*}
	according to its definition. Hence, for any $\iconstant \in \{1,\ldots, \ellsub{1} \}$,
	\begin{align*}
		|\zsln{\jconstant} - \xln{\iconstant}(\jconstant)| \le \max \{ |\xln{\iconstant}(\jconstant) - \xln{1}(\jconstant)|, \ldots, |\xln{\iconstant}(\jconstant) - \xln{\ellsub{1}}(\jconstant)| =  \xhln{\iconstant}(\jconstant), 
	\end{align*}
	following which 
	\begin{align*}
		& \sum_{\jconstant = 1}^{\lconstant} \xhln{\iconstant}(\jconstant) \ge \sum_{\jconstant = 1}^{\lconstant} |\xln{\iconstant}(\jconstant) - \zsln{\jconstant}|,~ \iconstant = \{1,\ldots,\ellsub{1}\}\\
		\Rightarrow &-\beta \sum_{\jconstant = 1}^{\lconstant} \xhln{\iconstant}(\jconstant) \le -\beta \sum_{\jconstant = 1}^{\lconstant} |\xln{\iconstant}(\jconstant) - \zsln{\jconstant}|,~ \iconstant = \{1,\ldots,\ellsub{1}\}.
	\end{align*}
	Subsequently,
	\begin{align*}
		\bt(\zsln{1}, \ldots, \zsln{\lconstant}) &= \sum_{\iconstant = 1}^{\ellsub{1}} \saln{\iconstant}  \exp\left(-\beta \sum_{\jconstant = 1}^{\lconstant} |\xln{\iconstant}(\jconstant) - \zsln{\jconstant}|\right) - \sum_{\iconstant = 1}^{\ellsub{2}} \bbln{\iconstant} \exp\left(-\beta  \sum_{\jconstant = 1}^{\lconstant}|\yln{\iconstant}(\jconstant) - \zsln{\jconstant}|\right) \\
		& \ge \sum_{\iconstant = 1}^{\ellsub{1}} \saln{\iconstant} \exp\left(-\beta \sum_{\jconstant = 1}^{\lconstant} \xhln{\iconstant}(\jconstant) \right) - \sum_{\iconstant = 1}^{\ellsub{2}} \bbln{\iconstant} \exp\left(-\beta  \sum_{\jconstant = 1}^{\lconstant}|\yln{\iconstant}(\jconstant) - \zsln{\jconstant}|\right) \\
		& \ge \sum_{\iconstant = 1}^{\ellsub{1}} \saln{\iconstant} \exp\left(-\beta \sum_{\jconstant = 1}^{\lconstant} \xhln{\iconstant}(\jconstant) \right) - \sum_{\iconstant = 1}^{\ellsub{2}} \bbln{\iconstant},
	\end{align*}
	which proves inequality \eqref{supp:lemma:5:eqn:1}.
\end{proof}

\subsection{Proof of Theorem \ref{theorem:1} when $\dd = 1$} \label{Proof of Theorem 1}

We can now prove our main result Theorem \ref{theorem:1}. To do so, we divide the proof of Theorem \ref{theorem:1} into univariate case ($\dd = 1$) and multivariate case ($\dd > 1$), which are provided in Theorem \ref{supp:theorem:1} and Theorem \ref{supp:theorem:2}, respectively.

\begin{theorem} \label{supp:theorem:1}
	Suppose $\bx = \{\xln{1}, \ldots, \xln{\nln{1}}\}$ and $\by = \{\yln{1}, \ldots, \yln{\nln{2}}\}$ are univariate real values. Assume $\xln{1}, \ldots, \xln{\mln{1}}$, $\yln{1}, \ldots, \yln{\mln{2}}$ are unobserved. Let $\kk$ denote the Laplacian kernel and define
	\begin{align*}
		\btln{1}(\z) &= \constant{1} \sum_{\jconstant = \mln{1} + 1}^{\nln{1}} \kernel{\z}{\xln{\jconstant}} - 	\constant{3} \sum_{\jconstant = \mln{2} + 1}^{\nln{2}} \kernel{\z}{\yln{\jconstant}},\\
		\btln{2}(\z) &= \constant{2} \sum_{\jconstant = \mln{2} + 1}^{\nln{2}} \kernel{\z}{\xln{\jconstant}} - 	\constant{3} \sum_{\jconstant = \mln{1} + 1}^{\nln{1}} \kernel{\z}{\yln{\jconstant}},
	\end{align*}
	where $\constant{1} = \frac{2}{\nln{1}(\nln{1}-1)}, \constant{2} = \frac{2}{\nln{2}(\nln{2}-1)}$ and $\constant{3} = \frac{2}{\nln{1}\nln{2}}$. Further, let
	\begin{align*}
		\bsln{1}:= \{0, \btln{1}(\xln{\mln{1}+1}),\ldots,  \btln{1}(\xln{\nln{1}}), \btln{1}(\yln{\mln{2}+1}),\ldots,  \btln{1}(\yln{\nln{2}})\},\\
		\bsln{2}:= \{0, \btln{2}(\xln{\mln{2}+1}),\ldots,  \btln{2}(\xln{\nln{2}}), \btln{2}(\yln{\mln{1}+1}),\ldots,  \btln{2}(\yln{\nln{1}})\}.
	\end{align*}	
	Then, the $\mmdu{\bx}{\by}$ using Laplacian kernel $\kk$ is bounded as follows
	\begin{align*}
		&\frac{\mln{1}(\mln{1} - 1)}{\nln{1}(\nln{1}-1)} + \frac{\mln{2}(\mln{2} - 1)}{\nln{2}(\nln{2}-1)} + \mln{1} \max \bsln{1} + \mln{2} \max \bsln{2} + \termtwo  >	\mmdu{\bx}{\by},\\
		&\mmdu{\bx}{\by} > \termtwo + \mln{1} \min \bsln{1} + \mln{2} \min \bsln{2} - \frac{2}{\nln{1}\nln{2}} \mln{1}\mln{2},
	\end{align*}
	where $\termtwo$ is defined in lemma \ref{lemma:1}.
\end{theorem}

\begin{proof}
	To start, in Lemma \ref{lemma:1}, it is shown that
	\begin{align*}
		& \mmdu{\bx}{\by} = \termone + \termtwo + \termthree + \termfour,
	\end{align*}
	where
	\begin{align*}
		\termone &= \constant{1} \sum_{\iconstant = 1}^{\mln{1}} \sum_{\jconstant = \iconstant + 1}^{\mln{1}} \kernel{\xln{\iconstant}}{\xln{\jconstant}} + \constant{2} \sum_{\iconstant = 1}^{\mln{2}} \sum_{\jconstant = \iconstant +1}^{\mln{2}} \kernel{\yln{\iconstant}}{\yln{\jconstant}} - \constant{3} \sum_{\iconstant = 1}^{\mln{1}} \sum_{\jconstant = 1}^{\mln{2}} \kernel{\xln{\iconstant}}{\yln{\jconstant}},\\
		\termtwo &= \constant{1} \sum_{\iconstant = \mln{1} + 1}^{\nln{1}-1} \sum_{\jconstant = \iconstant + 1}^{\nln{1}} \kernel{\xln{\iconstant}}{\xln{\jconstant}} + \constant{2}\sum_{\iconstant = \mln{2}+1}^{\nln{2}-1} \sum_{\jconstant = \iconstant +1}^{\nln{2}} \kernel{\yln{\iconstant}}{\yln{\jconstant}} - 	\constant{3}	\sum_{\iconstant = \mln{1} + 1}^{\nln{1}} \sum_{\jconstant = \mln{2} + 1}^{\nln{2}} \kernel{\xln{\iconstant}}{\yln{\jconstant}},\\
		\termthree & =  \constant{1}\sum_{\iconstant = 1}^{\mln{1}} \sum_{\jconstant = \mln{1} + 1}^{\nln{1}} \kernel{\xln{\iconstant}}{\xln{\jconstant}} - 	\constant{3} \sum_{\iconstant = 1}^{\mln{1}} \sum_{\jconstant = \mln{2} + 1}^{\nln{2}} \kernel{\xln{\iconstant}}{\yln{\jconstant}}, \\
		\termfour & = \constant{2} \sum_{\iconstant = 1}^{\mln{2}} \sum_{\jconstant = \mln{2} +1}^{\nln{2}} \kernel{\yln{\iconstant}}{\yln{\jconstant}} - 	\constant{3} \sum_{\iconstant = \mln{1} + 1}^{\nln{1}} \sum_{\jconstant = 1}^{\mln{2}} \kernel{\xln{\iconstant}}{\yln{\jconstant}},
	\end{align*}
	and $\constant{1} = \frac{2}{\nln{1}(\nln{1}-1)}, \constant{2} = \frac{2}{\nln{2}(\nln{2}-1)}$ and $\constant{3} = \frac{2}{\nln{1}\nln{2}}$.
	
	Following Equation \eqref{eqn:bounding termone}, it is provided
	\begin{align*}
		\frac{\mln{1}(\mln{1} - 1)}{\nln{1}(\nln{1}-1)} + \frac{\mln{2}(\mln{2} - 1)}{\nln{2}(\nln{2}-1)} > \termone > - \frac{2}{\nln{1}\nln{2}} \mln{1}\mln{2}.
	\end{align*}
	Further, using Lemma \ref{lemma:2}, it follows
	\begin{align*}
		\mln{1} \max \bsln{1} \ge \termthree \ge \mln{1} \min \bsln{1},
	\end{align*}
	and
	\begin{align*}
		\mln{2} \max \bsln{2} \ge \termfour \ge \mln{2} \min \bsln{2},
	\end{align*}
	which concludes our proof.	
\end{proof}

In order to compute the bounds of $\mmdu{\bx}{\by}$ for univariate samples with missing data using Theorem \ref{supp:theorem:1}, we need to compute $ \min \bsln{1}, \max \bsln{1}, \min \bsln{2}$ and $\max \bsln{2}$. Since $|\bsln{1}| = \nln{1} - \mln{1}$ and $|\bsln{2}| = \nln{2} - \mln{2}$. The computation complexity is of order $O(\nln{1} + \nln{2})$.

\subsection{Proof of Theorem \ref{theorem:1} when $\dd > 1$} \label{Proof of Theorem 2}

In Theorem \ref{supp:theorem:1}, the result of Theorem \ref{theorem:1} when $\dd = 1$ is provided. This section proves the result of  Theorem \ref{theorem:1} when $\dd > 1$. The final conclusion of this section is provided in Theorem \ref{supp:theorem:2}.  To obtain this result, we first introduce the following definition of the incomplete Laplacian kernel, which can be 
computed between two incomplete samples. 
\begin{definition} \label{def: incomplete kernel}
	For any $\x, \y \in \mathbb{R}^d$, let $\buln{\x, \y} \subset \{1,\ldots,\dd\}$ 
	be the index of dimensions for which either $\x$ or $\y$ have components that
	are not observed. Let $\kk$ denote the 
	Laplacian kernel with parameter $\beta$. Then, the incomplete Laplacian 
	kernel is defined as
	\begin{align*}
			\kernels{\x}{\y} = \exp\left(-\beta \sum_{\iconstant \in \{1,\ldots,\dd\} \setminus\left(\buln{\x} \cup \buln{\y}\right) } |\x(\iconstant) - \y(\iconstant)| \right).
		\end{align*}
\end{definition}
Subsequently, we have the following result.
\setcounter{lemma}{4}
\begin{lemma} \label{supp:lemma:3}
	Suppose $\x$ and $\y$ are two samples of $\mathbb{R}^d$ real values and assume not all dimensions of values of $\x$ and $\y$ are observed. Denote $\tonumber{\dd} = \{1,\ldots,\dd\}$. Let $\buln{\x}$ be a set includes all unobserved dimensions in $\x$ and $\buln{\y}$  be a set includes all unobserved dimensions in $\y$. Let $\kk^*$ denote the incomplete Laplacian kernel. Subsequently, it follows
	\begin{align*}
			\kernels{\x}{\y} \ge \kernel{\x}{\y} \ge 0.
		\end{align*} 
\end{lemma}

\begin{proof}
	According to the definition of the Laplacian kernel, we have
	\begin{align*}
			\kernel{\x}{\y} = \exp\left(-\beta \sum_{\lconstant=1}^{\dd}|\x(\lconstant) - \y(\lconstant)|\right) \ge 0.
		\end{align*}
	Also, notice that
	\begin{align*}
			\sum_{\lconstant \in \tonumber{\dd} \setminus\left(\buln{\x} \cup \buln{\y}\right) } |\x(\lconstant) - \y(\lconstant)| \le \sum_{\lconstant=1}^{\dd}|\x(\lconstant) - \y(\lconstant)|.
		\end{align*}
	Then, since $\beta > 0$,
	\begin{align*}
			-\beta \sum_{\iconstant \in \tonumber{\dd} \setminus\left(\buln{\x} \cup \buln{\y}\right) } |\x(\iconstant) - \y(\iconstant)| \ge - \beta \sum_{\lconstant=1}^{\dd}|\x(\lconstant) - \y(\lconstant)|.
		\end{align*}
	According to the definition of incomplete Laplacian kernel, we therefore have
	\begin{align*}
			\kernels{\x}{\y} &= \exp\left(-\beta \sum_{\lconstant \in \tonumber{\dd} \setminus\left(\buln{\x} \cup \buln{\y}\right) } |\x(\lconstant) - \y(\lconstant)| \right)\\
			& \ge \exp\left(-\beta \sum_{\lconstant=1}^{\dd}|\x(\lconstant) - \y(\lconstant)|\right) \\
			& = \kernel{\x}{\y},
		\end{align*}
	which finishes our proof.
\end{proof}

Recall that in Lemma \ref{lemma:1}, it is shown that $\mmdu{\bx}{\by}$ can be decomposed into four parts $\termone, \termtwo, \termthree$ and $\termfour$ with part $\termone$ includes incomplete samples only, $\termtwo$ includes complete samples only, and $\termthree$, $\termfour$ mixed with both complete and incomplete samples.

The Lemma \ref{supp:lemma:3} allows us to bound term $\termone$ by bounding any two incompletely observed samples separately. Let $\buln{\xln{\iconstant}}$ be a set including all unobserved dimensions in $\xln{\iconstant}$, where $1\le \iconstant \le \mln{1}$. Similarly, let $\buln{\yln{\iconstant}}$ be a set including all unobserved dimensions in $\yln{\iconstant}$, where $1\le \iconstant \le \mln{2}$. Subsequently, according to Lemma \ref{supp:lemma:3}, it can be seen that $\sum_{\iconstant = 1}^{\mln{1}} \sum_{\jconstant = \iconstant + 1}^{\mln{1}} \kernels{\xln{\iconstant}}{\xln{\jconstant}} \ge \sum_{\iconstant = 1}^{\mln{1}} \sum_{\jconstant = \iconstant + 1}^{\mln{1}} \kernel{\xln{\iconstant}}{\xln{\jconstant}} > 0$, $\sum_{\iconstant = 1}^{\mln{2}} \sum_{\jconstant = \iconstant +1}^{\mln{2}} \kernels{\yln{\iconstant}}{\yln{\jconstant}} \ge \sum_{\iconstant = 1}^{\mln{2}} \sum_{\jconstant = \iconstant +1}^{\mln{2}} \kernel{\yln{\iconstant}}{\yln{\jconstant}} > 0$, and $\sum_{\iconstant = 1}^{\mln{1}} \sum_{\jconstant = 1}^{\mln{2}} \kernels{\xln{\iconstant}}{\yln{\jconstant}} \ge \sum_{\iconstant = 1}^{\mln{1}} \sum_{\jconstant = 1}^{\mln{2}} \kernel{\xln{\iconstant}}{\yln{\jconstant}} > 0$. Then, according to the definition of $\termone$ in Lemma \ref{lemma:1}, it follows
\begin{align} \label{eqn:boundstermone:multi}
\constant{1}{\sum_{\iconstant = 1}^{\mln{1}}\sum_{\jconstant = \iconstant + 1}^{\mln{1}} \kernels{\xln{\iconstant}}{\xln{\jconstant}}} + \constant{2}{\sum_{\iconstant = 1}^{\mln{2}}\sum_{\jconstant = \iconstant + 1}^{\mln{2}} \kernels{\yln{\iconstant}}{\yln{\jconstant}}}  > \termone > -\constant{3} \sum_{\iconstant = 1}^{\mln{1}}\sum_{\jconstant = 1}^{\mln{2}} \kernels{\xln{\iconstant}}{\yln{\jconstant}},
\end{align} 
where $\constant{1} = \frac{2}{\nln{1}(\nln{1}-1)}, \constant{2} = \frac{2}{\nln{2}(\nln{2}-1)}$ and $\constant{3} = \frac{2}{\nln{1}\nln{2}}$. We now bound the terms $\termthree$ and $\termfour$. Recall that the term $\termthree$ can be rewritten as \eqref{rewrite:termthree:1} in both univariate and multivariate cases. For any $\xln{\iconstant}$ where $1\le \iconstant \le \mln{1}$, let $\buln{\xln{\iconstant}}$ be a set including all unobserved dimensions in $\xln{\iconstant}$ and let us introduce the function
\begin{align} \label{def:T_1}
	\btln{1}(\xln{\iconstant}(\lconstant); \lconstant \in \buln{\xln{\iconstant}}) = \constant{1} \sum_{\jconstant = \mln{1} + 1}^{\nln{1}} \kernel{\xln{\iconstant}}{\xln{\jconstant}} - 	\constant{3} \sum_{\jconstant = \mln{2} + 1}^{\nln{2}} \kernel{\xln{\iconstant}}{\yln{\jconstant}},~\iconstant \in \{1,\ldots,\mln{1}\}
\end{align}
where $\constant{1} = \frac{2}{\nln{1}(\nln{1}-1)}$ and $\constant{3} = \frac{2}{\nln{1}\nln{2}}$. Then, $\termthree$ can be computed as 
\begin{align*}
	\termthree = \sum_{\iconstant = 1}^{\mln{1}} \btln{1}(\xln{\iconstant}(\lconstant); \lconstant \in \buln{\xln{\iconstant}}).
\end{align*}
Subsequently, by combining Lemma \ref{supp:lemma:4} and Lemma \ref{supp:lemma:5}, we have the results of bounds of $\mmdu{\bx}{\by}$ when $\dd > 1$.

\begin{theorem} \label{supp:theorem:2}
	Suppose $\bx = \{\xln{1}, \ldots, \xln{\nln{1}}\}$ and $\by = \{\yln{1}, \ldots, \yln{\nln{2}}\}$ are samples of $\mathbb{R}^{\dd}$ real values. Assume $\xln{1}, \ldots, \xln{\mln{1}}$, $\yln{1}, \ldots, \yln{\mln{2}}$ are samples that are not observed completely. Let $\kk$ denote the Laplacian kernel and $\kk^*$ denote the incomplete Laplacian kernel defined in Definition \ref{def: incomplete kernel}. Further, denote
	\begin{align*}
		\xhln{\iconstant}(\lconstant) &= \max \{ |\xln{\iconstant}(\lconstant) - \xln{\mln{1}+1}(\lconstant)|, \ldots, |\xln{\iconstant}(\lconstant) - \xln{\nln{1}}(\lconstant)|, |\xln{\iconstant}(\lconstant) - \yln{\mln{2}+1}(\lconstant)|, \ldots, |\xln{\iconstant}(\lconstant) - \yln{\nln{2}}(\lconstant)|\}
	\end{align*}
	for any $\iconstant \in \{\mln{1}+1,\ldots,\nln{1}\}, \lconstant \in \{1,\ldots,\dd\}$; denote
	\begin{align*}
		\yhln{\iconstant}(\lconstant) &= \max \{ |\yln{\iconstant}(\lconstant) - \xln{\mln{1}+1}(\lconstant)|, \ldots, |\yln{\iconstant}(\lconstant) - \xln{\nln{1}}(\lconstant)|, |\yln{\iconstant}(\lconstant) - \yln{\mln{2}+1}(\lconstant)|, \ldots, |\yln{\iconstant}(\lconstant) - \yln{\nln{2}}(\lconstant)|\}
	\end{align*}
	for any $\iconstant \in \{\mln{2}+1,\ldots,\nln{2}\}, \lconstant \in \{1,\ldots,\dd\}$. Let
	\begin{align*}
		C_1 &= \sum_{\iconstant = 1}^{\mln{1}} \max\left\{0,  \constant{1}\sum_{\jconstant = \mln{1}+1}^{\nln{1}} \kernels{\xln{\iconstant}}{\xln{\jconstant} } - \constant{3}\sum_{\jconstant = \mln{2}+1}^{\nln{2}} \kernels{\xln{\iconstant}}{\yln{\jconstant} } \exp\left(-\beta \sum_{\lconstant \in \buln{\xln{\iconstant}}} \yhln{\jconstant}(\lconstant) \right)  \right\},\\
		C_2 & = \sum_{\iconstant = 1}^{\mln{1}} \min\left\{0,  \constant{1}\sum_{\jconstant = \mln{1}+1}^{\nln{1}} \kernels{\xln{\iconstant}}{\xln{\jconstant} } \exp\left(-\beta \sum_{\lconstant \in \buln{\xln{\iconstant}}} \xhln{\jconstant}(\lconstant) \right) - \constant{3}\sum_{\jconstant = \mln{2}+1}^{\nln{2}} \kernels{\xln{\iconstant}}{\yln{\jconstant} }  \right\}, \\
		C_3 & = \sum_{\iconstant = 1}^{\mln{2}} \max\left\{0,  \constant{2}\sum_{\jconstant = \mln{2}+1}^{\nln{2}} \kernels{\yln{\iconstant}}{\yln{\jconstant} } - \constant{3}\sum_{\jconstant = \mln{1}+1}^{\nln{1}} \kernels{\yln{\iconstant}}{\xln{\jconstant} } \exp\left(-\beta \sum_{\lconstant \in \buln{\yln{\iconstant}}} \xhln{\jconstant}(\lconstant) \right)  \right\}, \\
		C_4 & = \sum_{\iconstant = 1}^{\mln{2}} \min\left\{0,  \constant{2}\sum_{\jconstant = \mln{2}+1}^{\nln{2}} \kernels{\yln{\iconstant}}{\yln{\jconstant} } \exp\left(-\beta \sum_{\lconstant \in \buln{\yln{\iconstant}}} \yhln{\jconstant}(\lconstant) \right) - \constant{3}\sum_{\jconstant = \mln{1}+1}^{\nln{1}} \kernels{\yln{\iconstant}}{\xln{\jconstant} }  \right\},
	\end{align*}
	where $\constant{1} = \frac{2}{\nln{1}(\nln{1}-1)}, \constant{2} = \frac{2}{\nln{2}(\nln{2}-1)}$ and $\constant{3} = \frac{2}{\nln{1}\nln{2}}$.	Subsequently, the $\mmdu{\bx}{\by}$ using Laplacian kernel $\kk$ is bounded as follows
	\begin{alignat*}{2}
		&\constant{1}{\sum_{\iconstant = 1}^{\mln{1}}\sum_{\jconstant = \iconstant + 1}^{\mln{1}} \kernels{\xln{\iconstant}}{\xln{\jconstant}}} + \constant{2}{\sum_{\iconstant = 1}^{\mln{2}}\sum_{\jconstant = \iconstant + 1}^{\mln{2}} \kernels{\yln{\iconstant}}{\yln{\jconstant}}} + \termtwo + C_1 + C_3 > \mmdu{\bx}{\by},\\
		&\mmdu{\bx}{\by} > - \constant{3} \sum_{\iconstant = 1}^{\mln{1}}\sum_{\jconstant = 1}^{\mln{2}} \kernels{\xln{\iconstant}}{\yln{\jconstant}} + \termtwo + C_2 + C_4,
	\end{alignat*}
	where $\termtwo$ is defined in Lemma \ref{lemma:1}.
\end{theorem}

\begin{proof}
	In Lemma \ref{lemma:1}, it is proved that
	\begin{align*}
		\mmdu{\bx}{\by} = \termone + \termtwo + \termthree + \termfour,
	\end{align*}
	where
	\begin{align*}
		\termone &= \constant{1} \sum_{\iconstant = 1}^{\mln{1}} \sum_{\jconstant = \iconstant + 1}^{\mln{1}} \kernel{\xln{\iconstant}}{\xln{\jconstant}} + \constant{2} \sum_{\iconstant = 1}^{\mln{2}} \sum_{\jconstant = \iconstant +1}^{\mln{2}} \kernel{\yln{\iconstant}}{\yln{\jconstant}} - \constant{3} \sum_{\iconstant = 1}^{\mln{1}} \sum_{\jconstant = 1}^{\mln{2}} \kernel{\xln{\iconstant}}{\yln{\jconstant}},\\
		\termtwo &= \constant{1} \sum_{\iconstant = \mln{1} + 1}^{\nln{1}-1} \sum_{\jconstant = \iconstant + 1}^{\nln{1}} \kernel{\xln{\iconstant}}{\xln{\jconstant}} + \constant{2}\sum_{\iconstant = \mln{2}+1}^{\nln{2}-1} \sum_{\jconstant = \iconstant +1}^{\nln{2}} \kernel{\yln{\iconstant}}{\yln{\jconstant}} - 	\constant{3}	\sum_{\iconstant = \mln{1} + 1}^{\nln{1}} \sum_{\jconstant = \mln{2} + 1}^{\nln{2}} \kernel{\xln{\iconstant}}{\yln{\jconstant}},\\
		\termthree & =  \constant{1}\sum_{\iconstant = 1}^{\mln{1}} \sum_{\jconstant = \mln{1} + 1}^{\nln{1}} \kernel{\xln{\iconstant}}{\xln{\jconstant}} - 	\constant{3} \sum_{\iconstant = 1}^{\mln{1}} \sum_{\jconstant = \mln{2} + 1}^{\nln{2}} \kernel{\xln{\iconstant}}{\yln{\jconstant}}, \\
		\termfour & = \constant{2} \sum_{\iconstant = 1}^{\mln{2}} \sum_{\jconstant = \mln{2} +1}^{\nln{2}} \kernel{\yln{\iconstant}}{\yln{\jconstant}} - 	\constant{3} \sum_{\iconstant = \mln{1} + 1}^{\nln{1}} \sum_{\jconstant = 1}^{\mln{2}} \kernel{\xln{\iconstant}}{\yln{\jconstant}},
	\end{align*}
	and $\constant{1} = \frac{2}{\nln{1}(\nln{1}-1)}, \constant{2} = \frac{2}{\nln{2}(\nln{2}-1)}$ and $\constant{3} = \frac{2}{\nln{1}\nln{2}}$.
	
	According to inequality \eqref{eqn:boundstermone:multi}, 
	\begin{align} \label{supp:theorem:2:eqn:1} 
		\constant{1}{\sum_{\iconstant = 1}^{\mln{1}}\sum_{\jconstant = \iconstant + 1}^{\mln{1}} \kernels{\xln{\iconstant}}{\xln{\jconstant}}} + \constant{2}{\sum_{\iconstant = 1}^{\mln{2}}\sum_{\jconstant = \iconstant + 1}^{\mln{2}} \kernels{\yln{\iconstant}}{\yln{\jconstant}}}  > \termone > -\constant{3} \sum_{\iconstant = 1}^{\mln{1}}\sum_{\jconstant = 1}^{\mln{2}} \kernels{\xln{\iconstant}}{\yln{\jconstant}}.
	\end{align} 
	
	Notice that
	\begin{align*}
		\termthree & =  \sum_{\iconstant = 1}^{\mln{1}} \left(\constant{1} \sum_{\jconstant = \mln{1} + 1}^{\nln{1}} \kernel{\xln{\iconstant}}{\xln{\jconstant}} - 	\constant{3} \sum_{\jconstant = \mln{2} + 1}^{\nln{2}} \kernel{\xln{\iconstant}}{\yln{\jconstant}} \right).
	\end{align*}
	Let us introduce the function
	\begin{align*}
		\btln{1}(\xln{\iconstant}(\lconstant); \lconstant \in \buln{\xln{\iconstant}}) = \constant{1} \sum_{\jconstant = \mln{1} + 1}^{\nln{1}} \kernel{\xln{\iconstant}}{\xln{\jconstant}} - 	\constant{3} \sum_{\jconstant = \mln{2} + 1}^{\nln{2}} \kernel{\xln{\iconstant}}{\yln{\jconstant}},~\iconstant \in \{1,\ldots,\mln{1}\}.
	\end{align*}
	Subsequently, $\termthree$ can be computed as 
	\begin{align} \label{supp:theorem:2:eqn:2}
		\termthree = \sum_{\iconstant = 1}^{\mln{1}} \btln{1}(\xln{\iconstant}(\lconstant); \lconstant \in \buln{\xln{\iconstant}}).
	\end{align}
	Denote
	\begin{align*}
		\setchiln{\xln{\iconstant}} = \{\btln{1}(\xln{\iconstant}(\lconstant); \lconstant \in \buln{\xln{\iconstant}})| \xln{\lconstant} \in \{\xln{\mln{1}+1}(\lconstant), \ldots, \xln{\nln{1}}(\lconstant), \yln{\mln{2}+1}(\lconstant), \ldots, \yln{\nln{2}}(\lconstant) \}, \lconstant \in \buln{\xln{\iconstant}} \},
	\end{align*}	
	for any $\iconstant = 1,\ldots,\mln{1}$. Let $\tonumber{\dd} = \{1,\ldots,\dd\}$. Notice that
	\begin{align*}
		\btln{1}(\xln{\iconstant}(\lconstant); \lconstant \in \buln{\xln{\iconstant}}) &= \sum_{\jconstant = \mln{1} + 1}^{\nln{1}} \salnp{\iconstant,\jconstant} \exp\left(-\beta \sum_{\lconstant \in \buln{\xln{\iconstant}}} |\xln{\iconstant}(\lconstant) - \xln{\jconstant}(\lconstant)|\right) - \\
		&\sum_{\jconstant = \mln{2} + 1}^{\nln{2}} \bblnp{\iconstant,\jconstant} \exp\left(-\beta \sum_{\lconstant \in \buln{\xln{\iconstant}}} |\xln{\iconstant}(\lconstant) - \yln{\jconstant}(\lconstant)|\right),
	\end{align*}	
	where
	\begin{align*}
		\salnp{\iconstant,\jconstant} &= \constant{1} \exp\left(-\beta \sum_{\lconstant \in \tonumber{\dd} \setminus \buln{\xln{\iconstant}}} |\xln{\iconstant}(\lconstant) - \xln{\jconstant}(\lconstant)|\right),\\
		\bblnp{\iconstant,\jconstant} &=  \constant{3}\exp\left(-\beta \sum_{\lconstant \in \tonumber{\dd} \setminus\buln{\xln{\iconstant}}} |\xln{\iconstant}(\lconstant) - \yln{\jconstant}(\lconstant)|\right). 
	\end{align*}
	
	Then, according to Lemma \ref{lemma:4}, it follows
	\begin{align} \label{supp:theorem:2:eqn:3}
		\max\{0, \max \setchiln{\xln{\iconstant}}\} \ge  \btln{1}(\xln{\iconstant}(\lconstant); \lconstant \in \buln{\xln{\iconstant}}) \ge \min \{0, \min \setchiln{\xln{\iconstant}}\}, \iconstant = 1,\ldots,\mln{1}.
	\end{align}
	According to Lemma \ref{lemma:5}, 
	\begin{align*}
		&\min \setchiln{\xln{\iconstant}} \ge \sum_{\jconstant = \mln{1} + 1}^{\nln{1}} \salnp{\iconstant,\jconstant} \exp\left(-\beta \sum_{\lconstant \in \buln{\xln{\iconstant}}} \xhln{\iconstant}(\lconstant) \right) -
		\sum_{\jconstant = \mln{2} + 1}^{\nln{2}} \bblnp{\iconstant,\jconstant},\\
		&\max \setchiln{\xln{\iconstant}} \le \sum_{\jconstant = \mln{1}+1}^{\nln{1}} \salnp{\iconstant, \jconstant} - \sum_{\jconstant = \mln{2}+1}^{\nln{2}} \bblnp{\iconstant,\jconstant} \exp\left(-\beta \sum_{\lconstant \in \buln{\xln{\iconstant}}} \yhln{\jconstant}(\lconstant) \right).
	\end{align*}
	Notice that according to the definition of incomplete kernel in Definition \ref{def: incomplete kernel},
	\begin{align*}
		\salnp{\iconstant,\jconstant} &= \constant{1} \exp\left(-\beta \sum_{\lconstant \in \tonumber{\dd} \setminus \buln{\xln{\iconstant}}} |\xln{\iconstant}(\lconstant) - \xln{\jconstant}(\lconstant)|\right) = \constant{1}\kernels{\xln{\iconstant}}{\xln{\jconstant}},\\
		\bblnp{\iconstant,\jconstant} &=  \constant{3}\exp\left(-\beta \sum_{\lconstant \in \tonumber{\dd} \setminus\buln{\xln{\iconstant}}} |\xln{\iconstant}(\lconstant) - \yln{\jconstant}(\lconstant)|\right) = \constant{3} \kernels{\xln{\iconstant}}{\yln{\jconstant}}. 
	\end{align*}
	Hence, for any $\iconstant = 1,\ldots,\mln{1}$, 
	\begin{align*}
		&\min \setchiln{\xln{\iconstant}} \ge \sum_{\jconstant = \mln{1} + 1}^{\nln{1}} \constant{1}\kernels{\xln{\iconstant}}{\xln{\jconstant}} \exp\left(-\beta \sum_{\lconstant \in \buln{\xln{\iconstant}}} \xhln{\iconstant}(\lconstant) \right) -
		\sum_{\jconstant = \mln{2} + 1}^{\nln{2}} \constant{3} \kernels{\xln{\iconstant}}{\yln{\jconstant}},\\
		&\max \setchiln{\xln{\iconstant}} \le \sum_{\jconstant = \mln{1}+1}^{\nln{1}} \constant{1}\kernels{\xln{\iconstant}}{\xln{\jconstant}} - \sum_{\jconstant = \mln{2}+1}^{\nln{2}} \constant{3} \kernels{\xln{\iconstant}}{\yln{\jconstant}} \exp\left(-\beta \sum_{\lconstant \in \buln{\xln{\iconstant}}} \yhln{\jconstant}(\lconstant) \right).
	\end{align*}
	According to \eqref{supp:theorem:2:eqn:3}, we then have
	\begin{align*}
		&\btln{1}(\xln{\iconstant}(\lconstant); \lconstant \in \buln{\xln{\iconstant}}) \ge \min \{0, \min \setchiln{\xln{\iconstant}}\} \\
		& \ge \min \left\{0, \sum_{\jconstant = \mln{1} + 1}^{\nln{1}} \constant{1}\kernels{\xln{\iconstant}}{\xln{\jconstant}} \exp\left(-\beta \sum_{\lconstant \in \buln{\xln{\iconstant}}} \xhln{\iconstant}(\lconstant) \right) -
		\sum_{\jconstant = \mln{2} + 1}^{\nln{2}} \constant{3} \kernels{\xln{\iconstant}}{\yln{\jconstant}} \right\},
	\end{align*}	
	and
	\begin{align*}
		&\btln{1}(\xln{\iconstant}(\lconstant); \lconstant \in \buln{\xln{\iconstant}}) \le \max\{0, \max \setchiln{\xln{\iconstant}}\}\\
		&\le \max\left\{0, \sum_{\jconstant = \mln{1}+1}^{\nln{1}} \constant{1}\kernels{\xln{\iconstant}}{\xln{\jconstant}} - \sum_{\jconstant = \mln{2}+1}^{\nln{2}} \constant{3} \kernels{\xln{\iconstant}}{\yln{\jconstant}} \exp\left(-\beta \sum_{\lconstant \in \buln{\xln{\iconstant}}} \yhln{\jconstant}(\lconstant) \right)\right\}.
	\end{align*}
	Following \eqref{supp:theorem:2:eqn:2}, then
	\begin{align*}
		C_1 \ge \termthree \ge C_2. 
	\end{align*}
	Using the similar method, it can be prove that
	\begin{align*}
		C_3 \ge \termfour \ge C_4.
	\end{align*}
	Combining with \eqref{supp:theorem:2:eqn:3}, we conclude our result.
\end{proof}

For computing bounds of $\mmdu{\bx}{\by}$ using Theorem \ref{supp:theorem:2}, we need to compute $	\xhln{\iconstant}(\lconstant)$ for any $\iconstant \in \{\mln{1}+1,\ldots,\nln{1}\}, \lconstant \in \{1,\ldots,\dd\}$, and compute $\yhln{\iconstant}(\lconstant)$ for any  $\iconstant \in \{\mln{2}+1,\ldots,\nln{2}\}, \lconstant \in \{1,\ldots,\dd\}$. Notice that for any $\xhln{\iconstant}(\lconstant)$, we have
\begin{align*}
	&|\{ |\xln{\iconstant}(\lconstant) - \xln{\mln{1}+1}(\lconstant)|, \ldots, |\xln{\iconstant}(\lconstant) - \xln{\nln{1}}(\lconstant)|, |\xln{\iconstant}(\lconstant) - \yln{\mln{2}+1}(\lconstant)|, \ldots, |\xln{\iconstant}(\lconstant) - \yln{\nln{2}}(\lconstant)|\}| \\
	&= \nln{1} - \mln{1} + \nln{2} - \mln{2}.
\end{align*}
Hence, computing all $\xhln{\iconstant}(\lconstant)$ is of computational order at most $O(\dd\nln{1}(\nln{1} + \nln{2}))$. Similarly, computing all $\yhln{\iconstant}(\lconstant)$ is of computational complexity order at most $O(\dd\nln{2}(\nln{1} + \nln{2}))$. After computing $\xhln{\iconstant}(\lconstant)$ for any $\iconstant \in \{\mln{1}+1,\ldots,\nln{1}\}, \lconstant \in \{1,\ldots,\dd\}$, and $\yhln{\iconstant}(\lconstant)$ for any  $\iconstant \in \{\mln{2}+1,\ldots,\nln{2}\}, \lconstant \in \{1,\ldots,\dd\}$, $C_1, C_2, C_3$ and $C_4$ can be computed with computational complexity order at most $O((\nln{1}+\nln{2}))$. Overall, the computational complexity order of computing bounds of $\mmdu{\bx}{\by}$ using Theorem \ref{supp:theorem:2}, is $O((\nln{1}+\nln{2}))$. This concludes our proof for Theorem \ref{theorem:1}.T

\subsection{Proof of Theorem \ref{theorem:3} } \label{Proof of Theorem 3}

This section provides bounds of $p$-value of MMD test statistic based on the bounds of MMD derived in Section \ref{Bounding MMD with Missing Data}, and the permutation test.
\setcounter{theorem}{1}
\begin{theorem} \label{supp:theorem:3}
	Suppose $\bx = \{\xln{1}, \ldots, \xln{\nln{1}}\}$ and $\by = \{\yln{1}, \ldots, \yln{\nln{2}}\}$ are  samples of $\ddimension$ real value. Let $(\sigma^{(1)}, \ldots, \sigma^{(B)})$ be $B$ $\iid$ random samplings, each being random permutation of $\{1,\ldots,\nln{1} + \nln{2}\}$ and denoted as $\sigma^{(\iconstant)} = (\sigma^{(\iconstant)}(1),\ldots, \sigma^{(\iconstant)}(\nln{1}+\nln{2}))$, $\iconstant = 1,\ldots,B$. Subsequently, let $\zln{1} = \xln{1},\ldots, \zln{\nln{1}} = \xln{\nln{1}}, \zln{\nln{1}+1} = \yln{1},\ldots, \zln{\nln{1}+\nln{2}} = \yln{\nln{2}}$ and for any $\iconstant = 1,\ldots,B$, denote
	\begin{align*}
		\bx_{\sigma^{(\iconstant)}} = \{\zln{\sigma^{(\iconstant)}(1)},\ldots,\zln{\sigma^{(\iconstant)}(\nln{1})} \}, \by_{\sigma^{(\iconstant)}} = \{\zln{\sigma^{(\iconstant)}(\nln{1}+1)},\ldots,\zln{\sigma^{(\iconstant)}(\nln{1}+\nln{2})} \},
	\end{align*}
	and define $p$ according to equation \eqref{def:p value using permutations}. Suppose further
	\begin{align*}
		\mmdulower{\bx}{\by}\le \mmdu{\bx}{\by},~ \mmdu{\bx_{\sigma^{(\iconstant)}}}{\by_{\sigma^{(\iconstant)}}} \le \mmduupper{\bx_{\sigma^{(\iconstant)}}}{\by_{\sigma^{(\iconstant)}}},~\iconstant = 1,\ldots, B.
	\end{align*}
	Define
	\begin{align*}
		\upperp = \frac{1}{B + 1} \left(1 + \sum_{\iconstant = 1}^{B} \indicator{\mmduupper{\bx_{\sigma^{(\iconstant)}}}{\by_{\sigma^{(\iconstant)}}} \ge \mmdulower{\bx}{\by}} \right).
	\end{align*}
	Then, we must have $\upperp \ge \p$.
\end{theorem}

\begin{proof}
	Notice that for any $\iconstant \in \{1,\ldots,B\}$, we have
	\begin{align*}
		&\indicator{\mmdu{\bx_{\sigma^{(\iconstant)}}}{\by_{\sigma^{(\iconstant)}}} \ge \mmdu{\bx}{\by} } = 1\\
		\Rightarrow & \mmdu{\bx_{\sigma^{(\iconstant)}}}{\by_{\sigma^{(\iconstant)}}} \ge \mmdu{\bx}{\by}.
	\end{align*}
	Since 
	\begin{align*}
		\mmdu{\bx}{\by} \ge \mmdulower{\bx}{\by},~ \mmduupper{\bx_{\sigma^{(\iconstant)}}}{\by_{\sigma^{(\iconstant)}}} \ge \mmdu{\bx_{\sigma^{(\iconstant)}}}{\by_{\sigma^{(\iconstant)}}},~\iconstant = 1,\ldots, B,
	\end{align*}
	it then follows
	\begin{align*}
		&\mmduupper{\bx_{\sigma^{(\iconstant)}}}{\by_{\sigma^{(\iconstant)}}} \ge \mmdu{\bx_{\sigma^{(\iconstant)}}}{\by_{\sigma^{(\iconstant)}}} \ge \mmdulower{\bx}{\by} ,\\
		\Rightarrow & \indicator{\mmduupper{\bx_{\sigma^{(\iconstant)}}}{\by_{\sigma^{(\iconstant)}}} \ge \mmdulower{\bx}{\by}} = 1.
	\end{align*}
	for any $\iconstant \in \{1,\ldots,B\}$.
	
	If, for any $\iconstant \in \{1,\ldots,B\}$,
	\begin{align*}
		\indicator{\mmdu{\bx_{\sigma^{(\iconstant)}}}{\by_{\sigma^{(\iconstant)}}} \ge \mmdu{\bx}{\by} }  = 0,
	\end{align*}
	then we have
	\begin{align*}
		\indicator{\mmduupper{\bx_{\sigma^{(\iconstant)}}}{\by_{\sigma^{(\iconstant)}}} \ge \mmdulower{\bx}{\by}} = 0~\text{or}~1.
	\end{align*}
	Hence, for  any $\iconstant \in \{1,\ldots,B\}$, we have
	\begin{align*}
		\indicator{\mmduupper{\bx_{\sigma^{(\iconstant)}}}{\by_{\sigma^{(\iconstant)}}} \ge \mmdulower{\bx}{\by}} \ge \indicator{\mmdu{\bx_{\sigma^{(\iconstant)}}}{\by_{\sigma^{(\iconstant)}}} \ge \mmdu{\bx}{\by} } .
	\end{align*}
	According to the definitions of $\p$ and $\upperp$, we then must have $\upperp \ge \p$.
	
\end{proof}

\subsection{Proof of Theorem \ref{theorem:4}} \label{Proof of Theorem 4}

This section provides bounds of $p$-value of MMD test statistic based on the bounds of MMD derived in Section \ref{Bounding MMD with Missing Data}, and the normality approximation proposed in \cite{gao2023two}.

\begin{theorem} \label{supp:theorem:4}
	Following the notation and definitions in Proposition 10 in \cite{gao2023two}, let $\kk$ denote the Laplacian kernel. Assume $X_1, \ldots, X_{\mln{1}}$ and $Y_1, \ldots, Y_{\mln{2}}$ are not observed. For $1 \le s,t \le N$ let us further denote 
	\begin{align*}
		\overline{a}_{s,t}^{k} = \left\{ \begin{array}{ll}
			\kernels{X_s}{X_t} & 1\le s,t \le \n,  s \neq t\\
			\kernels{X_s}{Y_{t-n}} & 1\le s \le n < t \le N, s \neq t \\
			\kernels{X_t}{Y_{s-n}} & 1\le t \le n < s \le N, \\
			\kernels{Y_{s-n}}{Y_{t-n}} & n+1<s,t\le N,
		\end{array}\right.,\underline{a}_{s,t}^{k} = \left\{ \begin{array}{ll}
			0 & 1\le s \le \mln{1}, 1 \le t \le N, \\ 0 & 1 \le s \le N, \n + 1 \le t \le \n+ \mln{2}, \\1 & s = t \\
			a_{s,t}^{k} & \mbox{Otherwise}.
		\end{array}\right.
	\end{align*}
	Let $\overline{A}_{s,t}^{k*} = \overline{a}_{s,t}^{k} - \underline{a}_{\cdot,t}^{k} - \underline{a}_{s,\cdot}^{k} + \overline{a}_{\cdot,\cdot}^{k}$, where
	\begin{align*}
		\underline{a}_{\cdot,t}^{k} = \frac{1}{N-2} \sum_{\iconstant=1}^{N} \underline{a}_{\iconstant,t}^{k},~		\underline{a}_{s,\cdot}^{k} = \frac{1}{N-2} \sum_{\jconstant=1}^{N} \underline{a}_{s,\jconstant}^{k},~\overline{a}_{\cdot,\cdot}^{k} = \frac{1}{(N-1)(N-2)} \sum_{\iconstant,\jconstant=1}^{N} \overline{a}_{\iconstant,\jconstant}^{k}.
	\end{align*}
	Define
	\begin{align*}
		\overline{\mathcal{V}}_{n,m}^{k*} = \frac{1}{N(N-3)}\sum_{s\neq t} \left(\overline{A}_{s,t}^{k*}\right)^2 - \frac{1}{(N-1)(N-3)}.
	\end{align*}
	Then, we must have $ \overline{\mathcal{V}}_{n,m}^{k*}  \ge {\mathcal{V}}_{n,m}^{k*},$
	Further, suppose 
	\begin{align*}
		0 \le \mmdulower{\bx}{\by}\le \mmdu{\bx}{\by}, 
	\end{align*}
	we then have
	\begin{align*}
		\frac{\mmdulower{\bx}{\by}}{\sqrt{\constant{\n,\m}\overline{\mathcal{V}}_{n,m}^{k*}}} \le 	\frac{\mmdu{\bx}{\by}}{\sqrt{\constant{\n,\m}{\mathcal{V}}_{n,m}^{k*}}}.
	\end{align*}
\end{theorem}

\begin{proof}
	To start, recall that it is proved in Lemma \ref{supp:lemma:3} that
	\begin{align*}
		\kernels{\x}{\y} \ge \kernel{\x}{\y} > 0,
	\end{align*}
	where $\x,\y$ are potentially incompletely observed $\mathbb{R}^d$ samples.
	Hence, for any $1 \le s,t \le N$, we have
	\begin{align*}
		\overline{a}_{s,t}^{k} \ge {a}_{s,t}^{k} \ge \underline{a}_{\cdot,t}^{k},
	\end{align*}
	which then gives
	\begin{align*}
		\overline{A}_{s,t}^{k*} \ge {A}_{s,t}^{k*},
	\end{align*}
	where ${A}_{s,t}^{k*}$ is defined in Proposition 10 in  \cite{gao2023two} as
	\begin{align*}
		{A}_{s,t}^{k*} := {a}_{s,t}^{k} - {a}_{\cdot,t}^{k} - {a}_{s,\cdot}^{k} + {a}_{\cdot,\cdot}^{k}.
	\end{align*}
	Subsequently,
	\begin{align*}
		\overline{\mathcal{V}}_{n,m}^{k*}  \ge {\mathcal{V}}_{n,m}^{k*},
	\end{align*}
	where $ {\mathcal{V}}_{n,m}^{k*}$ is defined in Proposition 10 in  \cite{gao2023two} as
	\begin{align*}
		{\mathcal{V}}_{n,m}^{k*} = \frac{1}{N(N-3)}\sum_{s\neq t} \left({A}_{s,t}^{k*}\right)^2 - \frac{1}{(N-1)(N-3)}.
	\end{align*}
	Further, suppose
	\begin{align*}
		0 \le \mmdulower{\bx}{\by}\le \mmdu{\bx}{\by}.
	\end{align*}
	Then, we must have
	\begin{align*}
		\frac{\mmdulower{\bx}{\by}}{\sqrt{\constant{\n,\m}\overline{\mathcal{V}}_{n,m}^{k*}}} \le 	\frac{\mmdu{\bx}{\by}}{\sqrt{\constant{\n,\m}{\mathcal{V}}_{n,m}^{k*}}},
	\end{align*}
	which completes our proof.
\end{proof}

In \cite{gao2023two}, the $p$-value is defined 
as $\p = 1 - \Phi\left(\frac{\mmdu{\bx}{\by}}{\sqrt{\constant{\n,\m}{\mathcal{V}}_{n,m}^{k*}}}\right)$, 
where $\Phi$ denotes the cumulative density function of the 
standard normal distribution. In the presence of missing data, 
we can define 
$\upperp = 1 - \Phi\left(\frac{\mmdulower{\bx}{\by}}{\sqrt{\constant{\n,\m}\overline{\mathcal{V}}_{n,m}^{k*}}}\right)$. 
Then, according to Theorem \ref{supp:theorem:4}, 
we must have $\upperp \ge \p$ when $0 \le \mmdulower{\bx}{\by}$. 
When $\mmdulower{\bx}{\by} < 0$, it suggests that the $\mmdu{\bx}{\by}$ 
could potentially be smaller than 0. Hence, we then define $\underline{p} = 1$ 
in this case, i.e. the null hypothesis will not be rejected if $\mmdulower{\bx}{\by} < 0$.

\newpage
\section{Simulation Details} \label{simulationdetails}

\subsection{Computing Environment and Resources} \label{computingenviro}

The experiments were run on an high performance computing cluster with 325 compute nodes, each equipped with 2x AMD EPYC 7742 processors (128 cores, 1TB RAM per node). Each job utilized 32 cores and 64 GB of memory, with compute times varying from a couple of hours to 2-3 days, based on the samples sizes and dimensions of data of the job.

%

\subsection{Simulation Settings} \label{simulationsettings}

\paragraph{Median Heuristic.} For all the experiments in the section \ref{Experiments}, Laplacian kernel is used and median heuristic \cite{gretton2012kernel} will be used to compute the parameter $\beta$ of the kernel. This is a popular method of computing $\beta$, and it generally works well \cite{gretton2012kernel}. For case deletion and the proposed methods, only completely observed samples will be used for computing $\beta$. For mean and hot deck imputation, the $\beta$ will be computed using samples after imputations. 

\paragraph{Case deletion.} For univariate data, case deletion uses only the observed data for testing. For multivariate data, only data that are completely observed are used.

\paragraph{Mean imputation.} In univariate data, mean imputation replaces missing values with the mean of the observed samples in the same group, $\bx$ or $\by$. For multivariate data, the missing components of an incomplete sample are imputed with the mean of the observed components of that same sample.

\paragraph{Hot deck imputation.} In univariate data, hot deck imputation imputes missing values with randomly selected observed values (with replacement) from the same group, $\bx$ or $\by$. For multivariate data, the missing components of an incomplete sample are imputed with randomly selected observed components (with replacement) from that same sample.

\paragraph{Implementations of common missing data approaches.} After performing case deletion, mean imputation, or hot deck imputation, standard MMD two-sample testing with a Laplacian kernel is conducted, with parameter $\beta$ computed using the median heuristic. Permutation tests with $B=100$ are used to compute the $p$-value.

\paragraph{Implementations of proposed methods.} In Theorem \ref{theorem:3}, it is proved that the $p$-value of $\mmdu{\bx}{\by}$, denoted as $\p$, can be bounded by $\upperp$, which can be computed using the bounds of $\mmdu{\bx}{\by}$ in the presence of missing data. The bounds of $\mmdu{\bx}{\by}$ are computed according to Theorem \ref{supp:theorem:1} when $\dd = 1$ and according to Theorem \ref{supp:theorem:2} when $\dd > 1$. Permutation times $B = 100$ will be used. The null hypothesis is rejected when $\upperp$ is larger or equal than the significance level $\alpha$. In Section \ref{Experiments}, we denote this proposed method based on permutations as \textbf{MMD-Miss: Perm}.  

When the sample sizes $\nln{1}, \nln{2} \ge 100, \dd \ge 50$, the normality approximation proposed in \cite{gao2023two} may be suitable for computing the $p$-value of the MMD test. In Theorem \ref{theorem:4}, we provide the upper bound the $p$-value, denoted as $\upperp$, by first bounding the value of estimation of variance and combining with the lower bounds of $\mmdu{\bx}{\by}$ provided in Section \ref{Bounding MMD with Missing Data}.  The null hypothesis will then be rejected when $\upperp$ is larger or equal than the significance level $\alpha$. In Section \ref{Experiments}, we denote this proposed method based on normality approximation as \textbf{MMD-Miss: Normality}.  

\paragraph{Implementations for generating missing data.} The missingness mechanisms for each experiment are detailed in Section \ref{Experiments}. If there are not enough samples in $\bx$ or $\by$ satisfying the conditions to be missing, randomly missing will be introduced. For example, in Figure \ref{fig:1}, $s \in \{0,\ldots,0.2\}$ proportion of samples in $\bx$ will be selected and labeled as missing and the missingness mechanism is that only sample smaller than 0 in $\bx$ will be randomly selected. If the proportion of samples smaller than 0 in $\bx$, denoted as $a$, is smaller than a given $s$, then all samples in $\bx$ smaller than 0 will be selected and $\lfloor(s - a)\nln{1}\rfloor$ number of unselected samples in $\bx$ will be randomly selected.

\newpage
\subsection{Examples of MNIST dataset.} \label{Examples of MNIST dataset}

In this section, we present examples of MNIST images used in Section \ref{Experiments}.

\begin{figure}[h] 
	\begin{center} 
		\includegraphics[width=14cm]{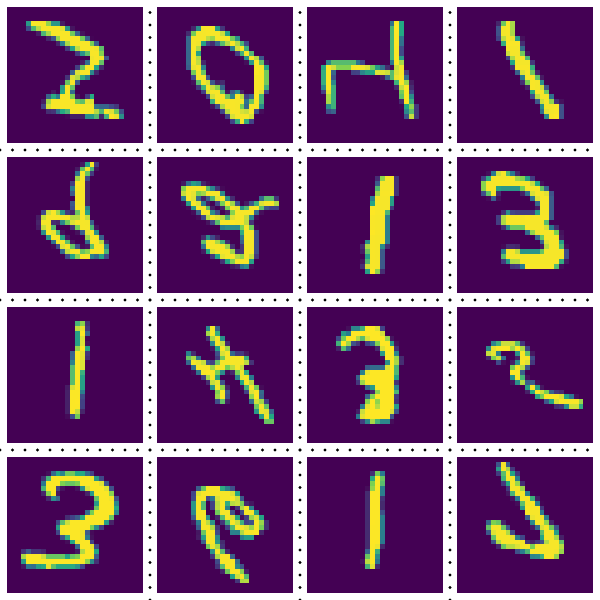}
	\end{center}
	\caption{Examples of MNIST images \cite{lecun1998gradient}. Each image has dimensions of $\dd = 28 \times 28$. The labels for the images, ordered from left to right and top to bottom, are 5, 0, 4, 1, 9, 2, 1, 3, 1, 4, 3, 5, 3, 6, 1, and 7.}  \label{fig:MNIST}
\end{figure}

\begin{figure}[h] 
	\begin{center} 
		\includegraphics[width=14cm]{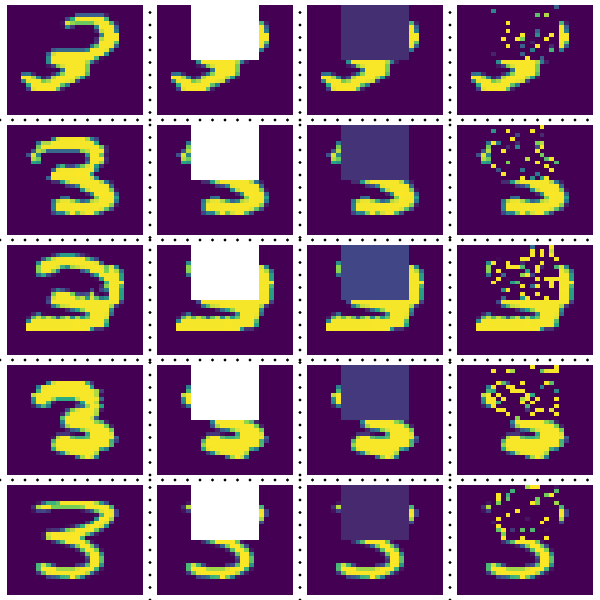}
	\end{center}
	\caption{ Examples of incompletely observed MNIST images, where white pixels represent unobserved values. Each image has dimensions of $\dd = 28 \times 28$ and is labeled as 3. The first column shows fully observed images; the second column shows images with upper half pixels missing (rows 1 to 14 and columns 8 to 21); the third column shows images after mean imputation; and the fourth column shows images after hot deck imputation. Each row presents different variations of the same original image.}  \label{fig:MNIST Incomplete}
\end{figure}

\clearpage
\section{Additional Experiments} \label{addtionalexperiments}

In this section, additional experiments are provided for investigating the power of proposed method. Different sample sizes and alternatives are considered compared to those in Section \ref{Experiments}.

\subsection{Univariate Data}
The missingness mechanism for producing the following two figures is the same as the mechanisms used to produce Figure \ref{fig:2}.
\begin{figure}[h] 
	\begin{center} 
		\includegraphics[width=14.1cm]{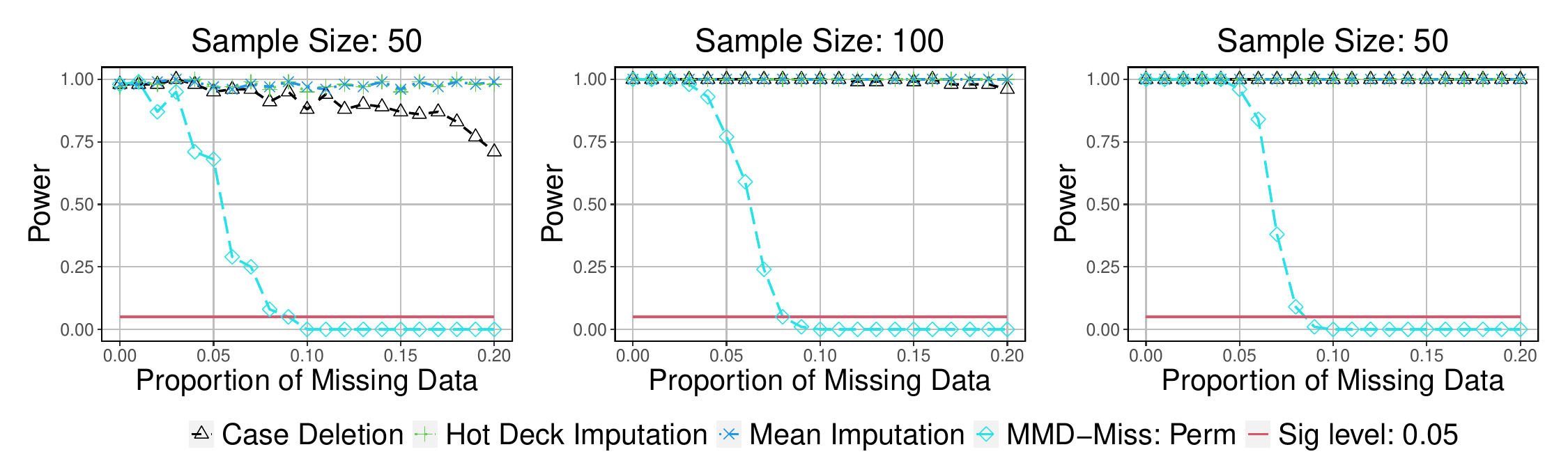}
	\end{center}
	\caption{The power of MMD-Miss and the three common 
		missing data approaches for univariate samples when data are missing not at random (MNAR). (Left): Sample size $\nln{1} = \nln{2} = 50$; (Middle): Sample size $\nln{1} = \nln{2} = 100$; (Right): Sample size $\nln{1} = \nln{2} = 200$. For all figures, significance level $\alpha = 0.05$, and alternative hypothesis $N(0,1)$ vs $N(1,1)$ are used. The plot values are the average times of rejecting the null hypothesis over 100 repetitions.}  \label{supp:fig:1}
\end{figure}

\begin{figure}[h] 
	\begin{center} 
		\includegraphics[width=14.1cm]{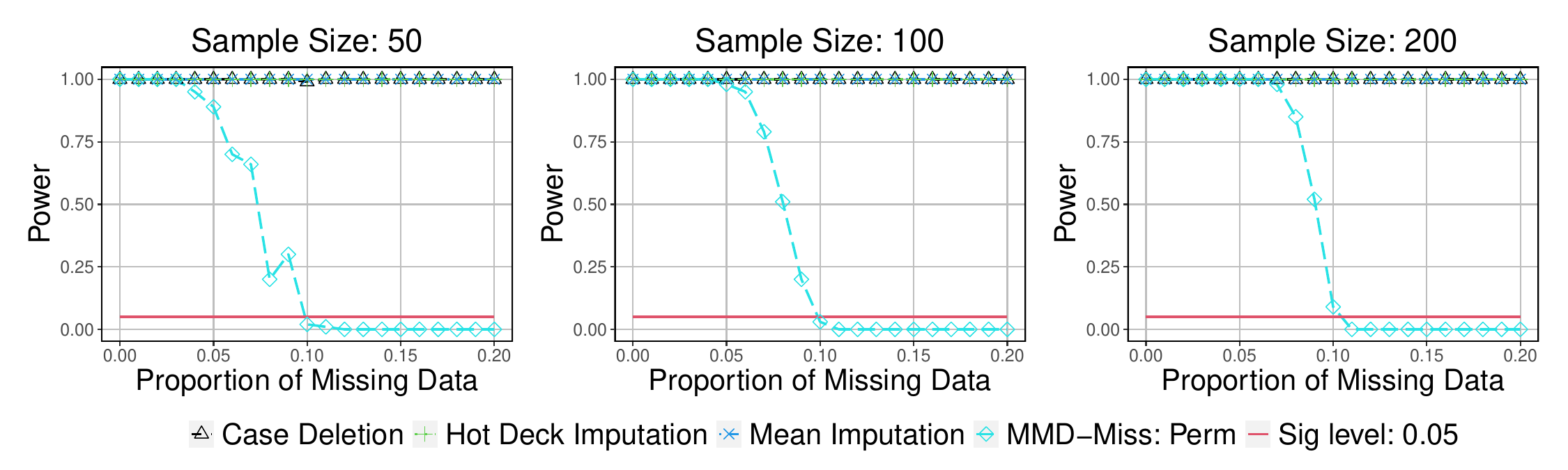}
	\end{center}
	\caption{The power of MMD-Miss and the three common 
		missing data approaches  for univariate samples when data are missing not at random (MNAR). (Left): Sample size $\nln{1} = \nln{2} = 50$; (Middle): Sample size $\nln{1} = \nln{2} = 100$; (Right): Sample size $\nln{1} = \nln{2} = 200$. For all figures, significance level $\alpha = 0.05$, and alternative hypothesis $N(0,1)$ vs $N(0,3)$ are used. The plot values are the average times of rejecting the null hypothesis over 100 repetitions.}  \label{supp:fig:2}
\end{figure}

\newpage
\subsection{Multivariate Data}
The missingness mechanism for producing the following two figures is the same as the mechanisms used to produce Figure \ref{fig:3}.
\begin{figure}[h] 
	\begin{center} 
		\includegraphics[width=14.1cm]{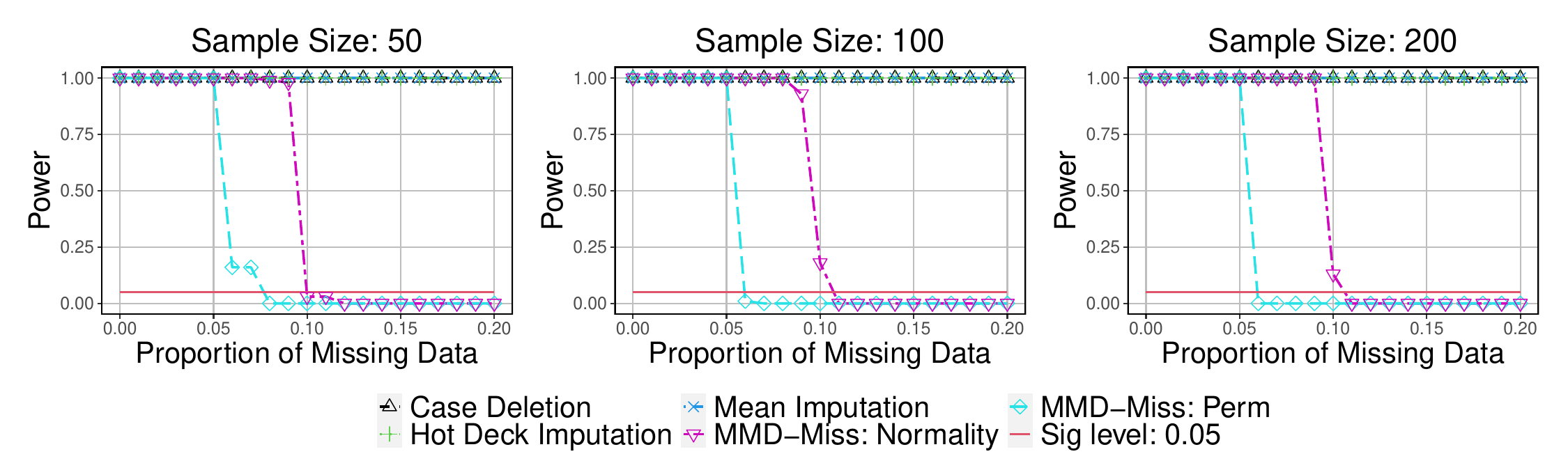}
	\end{center}
	\caption{The power of MMD-Miss and the three common 
		missing data approaches for multivariate samples $(\dd = 50)$ when data are missing not at random (MNAR). (Left): Sample size $\nln{1} = \nln{2} = 50$; (Middle): Sample size $\nln{1} = \nln{2} = 100$; (Right): Sample size $\nln{1} = \nln{2} = 200$. For all figures, significance level $\alpha = 0.05$, and alternative hypothesis $N((0,\ldots,0)^T,I_{50})$ vs $N((1,\ldots,1)^T,I_{50})$ are used. The plot values are the average times of rejecting the null hypothesis over 100 repetitions.}  \label{supp:fig:3}
\end{figure}

\begin{figure}[h] 
	\begin{center} 
		\includegraphics[width=14.1cm]{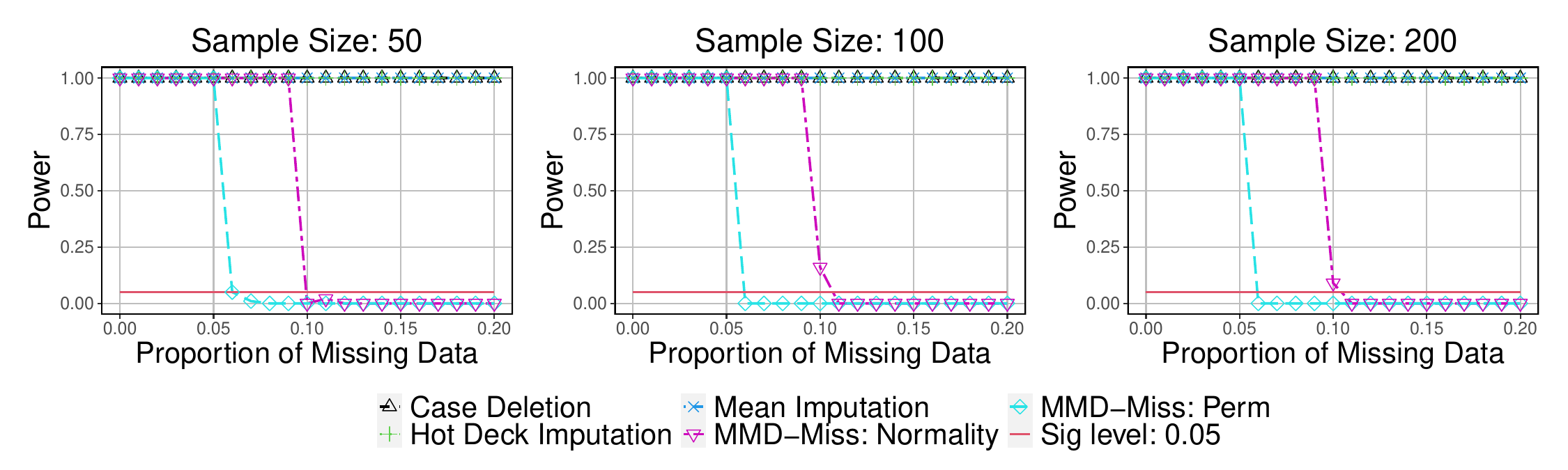}
	\end{center}
	\caption{The power of MMD-Miss and the three common 
		missing data approaches for multivariate samples $(\dd = 50)$ when data are missing not at random (MNAR). (Left): Sample size $\nln{1} = \nln{2} = 50$; (Middle): Sample size $\nln{1} = \nln{2} = 100$; (Right): Sample size $\nln{1} = \nln{2} = 200$. For all figures, significance level $\alpha = 0.05$, and alternative hypothesis $N((0,\ldots,0)^T,I_{50})$ vs $N((0,\ldots,0)^T,9I_{50})$ are used. The plot values are the average times of rejecting the null hypothesis over 100 repetitions.}  \label{supp:fig:4}
\end{figure}

\clearpage

\bibliographystyle{plain}
\bibliography{mybib}

\end{document}